\documentclass[12pt]{article}
\usepackage{times}
\usepackage[numbers]{natbib}
\usepackage{soul}
\usepackage{url}
\usepackage{hyperref}
\usepackage[utf8]{inputenc}
\usepackage[small]{caption}
\usepackage{graphicx}
\usepackage{amsmath}
\usepackage{booktabs}
\usepackage{algorithm}
\usepackage{algorithmic}
\usepackage{xcolor}
\usepackage{comment}
\definecolor{winered}{rgb}{0.6,0.1,0.1}
\hypersetup{menucolor=orange!40!black,filecolor=winered,urlcolor=green!40!black, pdfencoding=auto, psdextra,
colorlinks=true,citecolor=green!60!black,linkcolor=winered,
pagebackref=true}

\usepackage[utf8]{inputenc}
\usepackage{amsthm}
\usepackage{dsfont}
\usepackage{nicefrac}
\usepackage{amssymb}
\usepackage{graphicx}
\usepackage{mathtools}
\usepackage{tikz}
\usetikzlibrary{through}
\usepackage{enumerate}
\usepackage{pgfplots}

\newtheorem{theorem}{Theorem}[section]
\newtheorem{corollary}[theorem]{Corollary}
\newtheorem{definition}[theorem]{Definition}
\newtheorem{lemma}[theorem]{Lemma}

\newtheorem{proposition}[theorem]{Proposition}

\makeatletter
\newtheorem*{rep@theorem}{\rep@title}
\newcommand{\newreptheorem}[2]{%
\newenvironment{rep#1}[1]{%
 \def\rep@title{#2 \ref{##1}}%
 \begin{rep@theorem}}%
 {\end{rep@theorem}}}
\makeatother
\newreptheorem{theorem}{Theorem}
\newreptheorem{proposition}{Proposition}
\newreptheorem{lemma}{Lemma}
\newreptheorem{corollary}{Corollary}

\usepackage[sort&compress,nameinlink]{cleveref}
\crefname{algorithm}{Algorithm}{Algorithms}
\crefname{proposition}{Proposition}{Propositions}
\crefname{observation}{Observation}{Observations}
\crefname{corollary}{Corollary}{Corollaries}
\crefname{theorem}{Theorem}{Theorem}
\crefname{lemma}{Lemma}{Lemmas}
\crefname{claim}{Claim}{Claims}
\crefalias{AlgoLine}{line}
\crefname{example}{Example}{Examples}
\crefname{property}{Property}{Properties}
\crefname{condition}{Cond}{Cond}
\crefname{definition}{Definition}{Definitions}

\newcommand{\voters}[1]{R_A(c_#1)}

\newcommand{\leqref}[1]{\overset{\eqref{#1}}{\leq}}
\newcommand{\leqrefs}[2]{\overset{\eqref{#1},\eqref{#2}}{\leq}}

\newcommand{\distance}{d(c_o, c_w)}
\newcommand{\distancev}[1]{d(v, c_#1)}

\newcommand{\sumv}[2]{\sum_{v \in #1} d(v, c_#2)}
\newcommand{\sumvn}[2]{\sum_{v \notin #1} d(v, c_#2)}

\newcommand{\reals}{{{\mathbb{R}}}}

\allowdisplaybreaks[1]

\title{Approval-Based Elections and Distortion of Voting Rules}

\author{Grzegorz Pierczyński\\
  University of Warsaw\\
  Warsaw, Poland
  \and 
Piotr Skowron\\
  University of Warsaw\\
  Warsaw, Poland
}
\date{}

\begin{document}

\maketitle

\begin{abstract}
We consider elections where both voters and candidates can be associated with
points in a metric space and voters prefer candidates that are closer to those that are farther away.
It is often assumed that the optimal candidate is the one that minimizes the total distance to the voters.
Yet, the voting rules often do not have access to the metric space $M$ and only see preference rankings induced by $M$.
Consequently, they often are incapable of selecting the optimal candidate. 
The distortion of a voting rule measures the worst-case loss of the quality being the result of having access only to preference rankings.
We extend the idea of distortion to approval-based preferences. First, we compute the distortion of Approval Voting. Second,
we introduce the concept of acceptability-based distortion---the main idea behind is that the optimal candidate is the one
that is acceptable to most voters. We determine acceptability-distortion for a number of rules, including Plurality, Borda,
$k$-Approval, Veto, the Copeland's rule, Ranked Pairs, the Schulze's method, and STV.
\end{abstract}

\section{Introduction}\label{sec:intro}

We consider the classic election model: we are given a set of \emph{candidates}, a set of \emph{voters}---the voters have preferences over the candidates---and the goal is to select the \emph{winner}, i.e., the candidate that is (in some sense) most preferred by the voters. The two most common ways in which the voters express their preferences is (i) by ranking the candidates from the most to the least preferred one, or (ii) by providing \emph{approval sets}, i.e., subsets of candidates that they find acceptable. The collection of rankings (resp.\ approval sets), one for each voter, is called a ranking-based (resp.\ approval-based) profile.
There exist a plethora of rules that define how to select the winner based on a given preference profile, and comparing these election rules is one of the fundamental questions of the social choice theory~\cite{arr-sen-suz:b:handbook-of-social-choice}.

One such approach to comparing rules, proposed by \citet{pro-ros:c:distortion}, is based on the concept of \emph{distortion}.
Hereinafter, we explore its metric variant~\cite{ans-bha-elk-pos-sko:j:distortion}: the main idea is to assume that the voters and the candidates are represented by points in a metric space $M$ called the \emph{issue space}. The optimal candidate is the one that minimizes the sum of the distances to all the voters. However, the election rules do not have access to the metric space $M$ itself but they only see the ranking-based profile induced by $M$: in this profile the voters rank the candidates by their distance to themselves, preferring the ones that are closer to those that are farther. Since the rules do not have full information about the metric space they cannot always find optimal candidates. The distortion quantifies the worst-case loss of the utility being effect of having only access to rankings. Formally, the distortion of a voting rule is the maximum, over all metric spaces, of the following ratio: the sum of the distances between the elected candidate and the voters divided by the sum of the distances between the optimal candidate and the voters.

The concept of distortion is interesting, yet---in its original form---it only allows to compare ranking-based rules. In this paper we extend the distortion-based approach so that it captures approval preferences. In the first part of the paper we analyze the distortion of Approval Voting (AV), i.e., the rule that for each approval-based profile $A$ returns the candidate that belongs to the most approval sets from $A$. To formally define the distortion of AV one first needs to specify, for each metric space $M$, what is the approval-based profile induced by $M$. 
Here, we assume that each voter is the center of a certain ball and approves all the candidates within it.
We can see that each metric space induces a (possibly large) number of approval-based profiles---we obtain different profiles for different lengths of radiuses of the balls. This is different from ranking-based profiles, where (up to tie-breaking) each metric space induced exactly one profile. Thus, the distortion of AV might depend on how many candidates the voters decide to approve. Indeed, it is easy to observe that if each voter approves all the candidates, then the rule can pick any of them, which results in an arbitrarily bad distortion. On the other hand, by an easy argument we will show that for each metric space $M$ there exists an approval-based profile $A$ consistent with $M$, such that AV for $A$ selects the optimal candidate. In other words: AV can do arbitrarily well or arbitrarily bad, depending on how many candidates the voters approve.

Our first main contribution is that we fully characterize how the distortion of AV depends on the length of radiuses of approval balls. Specifically, we show that the distortion of AV is equal to 3, when the lengths of approval radiuses of the voters are all equal and such that the optimal candidate is approved by between $\nicefrac{1}{4}$ and $\nicefrac{1}{2}$ of the population of the voters (and this is the optimal distortion for the case of radiuses of equal length). The exact relation between the number of voters approving the optimal candidate and the distortion of AV is depicted in \Cref{fig:approval_voting_global_upper_bound}.

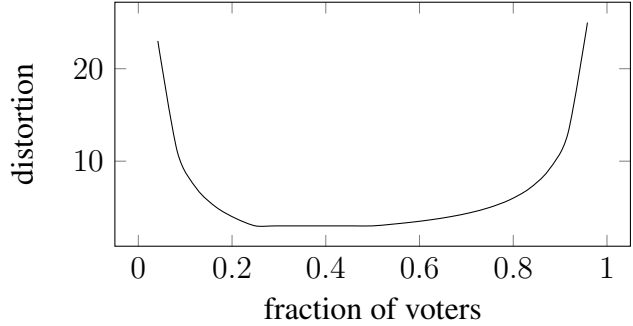
\begin{figure}[tp]
    \centering
    \begin{tikzpicture}[
              declare function={
    func(\x)= (\x<=0) * (500)   +
     and(\x>0, \x<=0.25) * ((1 - \x)/(\x))     +
     and(\x>0.25, \x<=0.5) * (3)     +
     and(\x>0.5,  \x<1) * ((2 - \x)/(1 - \x)) +
                (\x>=1) * (500);
  }
]
          \begin{axis}[ 
            xlabel={fraction of voters},
            ylabel={distortion},
            restrict y to domain=0:100,
            y=3.5,
          ] 
        \addplot[domain=0:1, smooth] {func(x)}; 
      \end{axis}
    \end{tikzpicture}
    \caption{The relation between the fraction of voters approving the optimal candidate and the distortion of AV, for the case when the approval radiuses of the voters have all equal lengths.}
    \label{fig:approval_voting_global_upper_bound}
\end{figure}

In the second part of the paper we explore the following related idea: assume that the goal of the election rule is not to select the candidate minimizing the total distance to the voters, but rather to pick the one that is acceptable for most of them. E.g., AV perfectly implements this idea. A natural question is how good are ranking-based rules with respect to this criterion. To answer this question we introduce a new concept of \emph{acceptability-based distortion} (in short, \emph{ab-distortion}). We assume that each metric space, apart from the points corresponding to the voters and candidates, contains acceptability balls---one for each voter (as before, each voter is the center of the corresponding ball). The optimal candidate is the one that belongs to the most acceptability balls, and the ab-distortion distortion measures the normalized difference between the numbers of balls to which the elected and the optimal candidates belong. The ab-distortion is a real number between 0 and~1, where~0 corresponds to selecting the optimal candidate and~1 is the worst possible value~\footnote{The reader might wonder why we define the ab-distortion as a difference rather than as a ratio (as it is done for the classic definition of the distortion). Indeed, we first used the ratios in our definition, but then it was very easy to construct instances where any rule had the distortion of $+\infty$. Further, we found that these results do not really speak of the nature of the rules but rather are artifacts of the used definition. Consequently, we found that the considering the difference gives more meaningful results.}.

Among the ranking-based rules that we consider in this paper, the best (and the optimal) ab-distortion is attained by Ranked Pairs and the Schulze's method. It is an open question, whether they are the only natural rules with this property. It is worth mentioning, that its ab-distortion is closely related to the size of the Smith set, so in case it is small (in particular, when the Condorcet winner exists) these rules have even better ab-distortion. We have found an interesting result for the Copeland's rule. Although in case of classic (distance-based) distortion most Condorcet rules are equally good, this is no longer the case when acceptability is the criterion we primarily care about. The ab-distortion of the Copeland's rule is equal to 1, which is the worst possible value. This rule is optimal only if the Condorcet winner exists (e.g. when the metric space is one-dimensional). The distortion of scoring rules (Plurality, Borda, Veto, k-approval) is significantly worse that for Ranked Pairs. An another surprising result is the distortion of STV---while this rule is known to achieve a very good distance-based distortion, its ab-distortion is even worse than for Plurality (denoting the number of candidates as $m$, STV and Plurality achieve the ab-distortion of $\frac{2^m-1}{2^m}$ and $\frac{m-1}{m}$, respectively). In case of all these rules the worst-case instances were obtained in one-dimensional Euclidean metric spaces. Our results are summarized in \Cref{tab:summary}.

\begin{table}
\centering
\begin{tabular}{l l l}
\toprule
\textbf{Rule} & \textbf{Distance-based distortion ($[1;+\infty]$)} & \textbf{Ab-distortion ($[0;1]$)}\\
\midrule
Each rule & $\geq 3$ & $\geq \frac{\ell-1}{\ell}$ for $\ell > 1$\\
& & $\geq \frac{1}{2}$ otherwise\\
\midrule
Plurality & $2m-1$ & $\frac{m-1}{m}$\\
\midrule
Borda & $2m-1$ & $\frac{m-1}{m}$\\
\midrule
k-approval & $\infty$ & $1$\\
\midrule
Veto & $\infty$ & $1$\\
\midrule
Copeland & $5$ & $1$ for $\ell > 1$\\
& & $\frac{1}{2}$ otherwise\\
\midrule
Ranked Pairs & $5$ & $\frac{\ell-1}{\ell}$ for $\ell > 1$\\
& & $\frac{1}{2}$ otherwise\\
\midrule
Schulze's Rule & $5$ & $\frac{\ell-1}{\ell}$ for $\ell > 1$\\
& & $\frac{1}{2}$ otherwise\\
\midrule
STV & $O(\ln m)$ & $\frac{2^{m-1}-1}{2^{m-1}}$\\
\bottomrule
\end{tabular}
\caption{The comparison of the distortion for various ranking-based rules. The results in the left column (for the distance-based distortion) are known in the literature. The results for ab-distortion are new to this paper; here, $m$ denotes the number of the candidates and $\ell$ is the size of the Smith set.}
\label{tab:summary}
\end{table}

\section{Preliminaries}\label{Preliminaries}
For each set $S$ by $2^S$ and $\Pi(S)$ we denote, respectively, the powerset of $S$ and the set of all linear orders over $S$. By $S^\complement$ we denote the complement set of $S$, and by $S^*$---the set of all vectors with the elements from $S$.
For each two sets $S_1, S_2$ and a function $f\colon S_1 \to 2^{S_2}$ by $R_f\colon S_2 \to 2^{S_1}$ we denote the function defined as follows:
\begin{align*}
    \forall y \in S_2 \quad R_f(y) = \{x \in S_1\colon y \in f(x)\}
\end{align*}

For convenience we assume that $[-\infty;+\infty]$ denotes the affinely extended real number system (the set of real numbers $\mathbb{R}$ with additional symbols $+\infty$, and $-\infty$). We take the following convention for arithmetical operations:
\begin{align*}
    \forall a \in \mathbb{R}\quad\frac{a}{\pm\infty}=0 \qquad \forall a \in (0;+\infty]\quad\frac{\pm a}{0}=\pm\infty \text{.}
\end{align*}
Expressions $\frac{0}{0}$, $\frac{\pm\infty}{\pm\infty}$, $0\cdot\pm\infty$ and $\pm\infty-\pm\infty$ are undefined.

\subsection{Our Metric Model}
An \emph{election instance} is a tuple $(N, C, d, \lambda)$, where $N = \{1, 2, \ldots, n\}$ is the set of \emph{voters}, $C=\{c_1, c_2, \ldots, c_m\}$, is the set of \emph{candidates}, $d \colon (N \cup C)^2 \to \mathbb{R}$ is a \emph{distance function} ($d$ allows us to view the candidates and the voters as points in a pseudo-metric space), and $\lambda\colon N \to 2^C$ is an \emph{acceptability function}, mapping each voter $i \in N$ to a subset of candidates that $i$ finds acceptable. We assume that $\lambda$ is \emph{nonempty}, i.e., for each $i \in N$, $\lambda(i) \neq \emptyset$, and that is \emph{local consistent}---for each $i \in N$, $c_a, c_b \in C$, if $c_a \in \lambda(i)$ and $d(i, c_b) \leq d(i, c_a)$, then $c_b \in \lambda(i)$. Often we will also require that $\lambda$ satisfies a stronger condition, called \emph{global consistency}---for each $i, j \in N$, $c_a, c_b \in C$, if $c_a \in \lambda(i)$ and $d(j, c_b) \leq d(i, c_a)$, then $c_b \in \lambda(j)$.
Intuitively, local-consistency means that for each voter $i \in N$ we can associate $\lambda(i)$ with a ball with the center at the point of this voter. A voter $i \in N$ considers a candidate $c_j$ to be acceptable for him, $c_j \in \lambda(i)$, if and only if $c_j$ lies within the ball. Such a ball will be further called the \emph{acceptability ball} and its radius---the \emph{acceptability radius}.
Then, global consistency can be interpreted as an assumption that all the acceptability radiuses have equal lengths.

We will sometimes slightly abuse the notation: by saying that an instance satisfies local (global) consistency we will mean that the acceptability function in the instance satisfies the respective property.  

By $\mathbb{I}$, we denote the set of all election instances. Since issue spaces are often argued to be Euclidean spaces with small numbers of dimensions, we additionally introduce the following notation: for each $k \in \mathbb{N}$ let $\mathbb{E}^k$ denote the set of all the instances where the elements of $N$ and $C$ are associated with points from $\mathbb{R}^k$, and $d$ is the Euclidean distance.

\subsection{Preference Representation}

In most cases, it is difficult for the voters to explicitly position themselves in the issue space, and often even the space itself is unknown. Therefore, we will consider voting rules that take as inputs preference profiles induced by election instances, instead of instances themselves. We consider two classic approaches to represent preferences.

\begin{description}
\item[Ranking-based profiles.] A \emph{ranking-based profile} induced by an election instance $I=(N, C, d, \lambda)$ is the function $\succcurlyeq_I\colon N \to \Pi(C)$, mapping each voter to a linear order over $C$ such that for all $i \in N$ and all $c_a, c_b \in C$ if $d(i, c_a) < d(i, c_b)$ then $c_a \succcurlyeq_i c_b$. For each voter $i \in N$, the relation ${\succcurlyeq_I}(i)$ (for convenience also denoted as $\succcurlyeq_i$, whenever the instance is clear from the context) is called the \emph{preference order} of $i$. If for some $c_x, c_y \in C$ it holds that $c_x \succcurlyeq_i c_y$, we say that $i$ prefers $c_x$ over~$c_y$.

\item[Approval-based profiles.] An \emph{approval-based profile} of an election instance $I=(N, C, d, \lambda)$ is a locally consistent acceptability function $A_I\colon N \to 2^C$. We say that a candidate $c_x$ is \emph{approved} by a voter $i \in N$ if $c_x \in A(i)$. 
We will say that the approval-based profile is \emph{truthful} if for all $i \in N$ it holds that $A(i) = \lambda(i)$.
\end{description}

Let us introduce some additional useful notation. Let $P\colon C^* \to N$ be a function mapping vectors of distinct candidates to sets of voters as follows:
\begin{equation*}
    P((c_{i_1}, c_{i_2}, ..., c_{i_k})) = \{v \in N : c_{i_1} \succcurlyeq_v c_{i_2} \succcurlyeq_v \ldots \succcurlyeq_v c_{i_k} \}
\end{equation*}

For convenience, we will write $P(c_{i_1}, c_{i_2}, ..., c_{i_k})$ instead of $P((c_{i_1}, c_{i_2}, ..., c_{i_k}))$\footnote{It will always be clear from the context whether in the inscription $P(x)$, $x$ should be interpreted as a vector or as a candidate.}. Note that for all $c_a, c_b \in C$ we have $P(c_a, c_b) \cap P(c_b, c_a) = \emptyset$ and $P(c_a, c_b) \cup P(c_b, c_a) = N$.

We say that $c_a$ \emph{dominates} $c_b$ if $|P(c_a, c_b)| > \frac{n}{2}$ and that $c_a$ \emph{weakly dominates} $c_b$ if $|P(c_a, c_b)| \geq \frac{n}{2}$. We say that a candidate $c_x$ \emph{Pareto-dominates} a candidate $c_y$ if there holds that $|P(c_x, c_y)| = n$. A candidate $c_y$ is \emph{Pareto-dominated} if there exists a candidate $c_x$ who Pareto-dominates $c_y$.

\subsection{Definitions of Voting Rules}

An \emph{election rule} (also referred to as a \emph{voting rule}) is a function mapping each preference profile to a set of tied \emph{winners}. We distinguish \emph{ranking-based} rules---taking ranking-based profiles as arguments, and approval-based rules---defined analogously. 
Among approval-based rules, we focus on \emph{Approval Voting} (AV)---the rule that selects those candidates that are approved by most voters.
In the remaining part of this subsection we recall definitions of the ranking-based rules that we study in this paper.

\paragraph{Positional scoring rules}
For a given vector $\vec{s} = (\alpha_1, \alpha_2, \alpha, \alpha_m)$, the scoring rule implemented by $\vec{s}$ works as follows.
A candidate $c_a$ gets $\alpha_i$ points from each voter $j$ who puts $c_a$ in the $i$th position in $\succcurlyeq_j$. The rule elects the candidates whose total number of point, collected from all the voters, is maximal. Some well-known scoring rules which we will study in the further part of this work are the following:
\begin{description}
    \item \textbf{Plurality:} $\vec{s} = (1, 0, \ldots, 0)$,
    \item \textbf{Veto:} $\vec{s} = (1,\ldots, 1, 0)$,
    \item \textbf{Borda:} $\vec{s} = (m-1, m-2, \ldots, 1, 0)$,
    \item \textbf{k-approval:} $\vec{s} = (\underbrace{1, \ldots, 1}_k, 0, \ldots, 0)$ (for $1 \leq k \leq m$).
\end{description}

\paragraph{The Copeland's Rule}
The \emph{Copeland's rule} elects candidates $c_w$ who dominate at least as many candidates as any other candidate. More formally, a candidate $c_w$ is a winner if and only if:
\begin{equation*}
    \forall c_x \in C \quad |\{c\colon |P(c_w, c)| > \frac{n}{2}\}| \geq |\{c\colon |P(c_x, c)| > \frac{n}{2}\}|
\end{equation*}

\paragraph{Ranked Pairs}
\emph{Ranked Pairs} works as follows: first we sort the pairs of candidates $(c_i, c_j)$ in the descending order of the values $|P(c_i, c_j)|$. Then, we construct a graph~$G$ where the vertices correspond to the candidates. We start with the graph with no edges; then we iterate over the sorted list of pairs---for each pair $(c_i, c_j)$ we add an edge from $c_i$ to $c_j$ unless there is already a path from $c_j$ to $c_i$ in $G$. If such a path exists, we simply skip this pair.
Clearly, the so-constructed graph $G$ is acyclic. The source nodes of $G$ are the winners. 

\paragraph{The Schulze's Rule}
The \emph{Schulze's rule} works as follows: let the \emph{beatpath} of length $k$ from candidate $c_a$ to $c_b$ be a sequence of candidates $c_{x_1}, c_{x_2},\ldots, c_{x_{k-1}}$ such that $c_a$ dominates $c_{x_1}$, $c_{x_{k-1}}$ dominates $c_b$ and for each $i \in \{1,\ldots, k-2\}$, $c_{x_i}$ dominates $c_{x_{i+1}}$. Let the \emph{strength} of the beatpath be the minimum of values $P(c_a, c_{x_1})$, $P(c_{x_1}, c_{x_2}), \ldots, P(c_{x_{k-1}}, c_b)$. By $p[c_a, c_b]$ we denote the maximum of strenghts of all beatpaths from $c_a$ to $c_b$. Candidate $c_w$ is the winner if and only if for each candidate $c$ it holds that $p[c_w, c] \geq p[c, c_w]$.



\paragraph{STV}
\emph{Single Transferable Vote (STV)} works iteratively as follows: if there is only one candidate, elect this candidate. Otherwise, eliminate the candidate who has the least points according to the Plurality rule and repeat the algorithm. 

\medskip
Note that the aforementioned rules are irresolute by definition. Further, we did not specify the tie-breaking rule used when sorting edges in Ranked Pairs and when eliminating candidates in STV. We will make all these rules resolute by using the lexicographical tie-breaking rule, denoted by $\succcurlyeq_{lex}$.

\subsection{Measuring the Quality of Social Choice}
In this section we formalize the concept of distortion that, on the intuitive level, we already introduced in \Cref{sec:intro}.


\paragraph{Distance-based approach}
A natural idea to relate the quality of a candidate $c$ with the sum of the distances from this candidate to all the voters. The lower this sum is, the higher the quality. 
Following this intuition, the \emph{distortion} of a voting rule~$\varphi$ in instance $I\in \mathbb{I}$, is defined as follows (below, $c_o$ denotes the optimal candidate for $I$):
\begin{equation*}
    D_I(\varphi) = \max_{p \in \mathcal{P}_I}\frac{ \sum_{i \in N} d(i, \varphi(p))}{ \sum_{i \in N} d(i, c_o)}\text{,}
\end{equation*}
where $\mathcal{P}_I$ is the set of profiles induced by $I$ (either ranking or approval, depending on the domain of $\varphi$).
$D(\varphi) \in [1;+\infty]$.

This approach can be applied to any rule discussed so far. For ranking-based rules it has already been widely studied in the literature, hence in the further part we will focus on AV.

\paragraph{Acceptability-based approach}
Now we present an alternative way to measure the quality of candidates, based on the acceptability function. Intuitively, the more voters a candidate $c$ is acceptable for, the higher his quality. Besides, we would like the maximal possible quality not to depend on the number of voters. 
Therefore, we define the \emph{acceptability-based distortion} (\emph{ab-distortion}, in short) of a voting rule $\varphi$ in instance $I\in \mathbb{I}$ as the following expression:
\begin{equation*}
    D_I(\varphi) = \max_{p\in \mathcal{P}_I} \frac{R_\lambda(c_o) - R_\lambda(\varphi(p))}{n}\text{,}
\end{equation*}
where $\mathcal{P}_I$ is the set of profiles induced by $I$ (either ranking-based or approval-based, depending on the domain of $\varphi$).
Clearly, the ab-distortion is always a value from $[0;1]$. By definition, Approval Voting always elects an optimal candidate in terms of ab-distortion. Thus, we will consider our acceptability-based measure only for ranking-based rules.

\medskip
Let $E$ be an expression that can depend on characteristics of an instance (e.g., on the number of candidates, or size of the Smith set). We say that the (acceptability-based) distortion of a rule $\varphi$ is $E$, if for each instance $I$, $D_I(\varphi) \leq E$ and for each $E$ there is an instance $I$ with $D_I(\varphi) = E$.


\section{Distortion of Approval Voting}

In this section we analyze the distance-based distortion of Approval Voting (AV)---hereinafter we denote AV by $\varphi_{AV}$.

We start by showing that in the most general case, if we do not make any additional assumptions about the acceptability function, the distortion of AV can be arbitrarily bad.  

\begin{proposition}\label[proposition]{prop:approval_general_worst_distortion}
There exists an instance $I \in \mathbb{E}^1$ such that $D_I(\varphi_{AV}) = +\infty$.
\end{proposition}

This result is rather pessimistic. However, one could ask a somehow related question---does there for each instance $I$ always exist an approval profile consistent with $I$ that would result in a good distortion? In contrast to \Cref{prop:approval_general_worst_distortion}, here the answer is much more positive.

\begin{proposition}\label[proposition]{prop:approval_general_best_distortion}
For each instance $I \in \mathbb{I}$, there is an approval based profile $p$ consistent with $I$ such that $\varphi(p)$ is the optimal candidate (minimizing the total distance to voters).
\end{proposition}

\Cref{prop:approval_general_worst_distortion,prop:approval_general_best_distortion} show that for each metric space $M$ there always exists two approval profile $A_1, A_2$ consistent with $M$ such that for $A_1$ AV selects the worst possible candidate, and for $A_2$ it selects the optimal one---since $A_1$ and $A_2$ are both consistent with $M$, they only differ in the sizes of approval balls. This formally shows that the performance of AV strongly depends on how many candidates the voters decide to approve. Below, we provide our main result of this section---assuming that all the acceptability balls have radiuses of the same length, we show the exact relation between this length of approval radiuses and the distance-based distortion of AV. In particular, we show that the best approval radius is such that the optimal candidate is approved by between $\nicefrac{1}{4}$ and $\nicefrac{1}{2}$ fraction of all the voters.  

\begin{definition}
An approval-based profile $A$ induced by an instance $I$ is \emph{$p$-efficient} for $p \in [0;1]$ if $R_A(c_o)=pn$.
\end{definition}

In words, a profile is \emph{p-efficient} if the number of voters who approve the optimal candidate is the $p$ fraction of $n$.

\begin{theorem} \label{approval_voting_global_upper_bound}
    For each globally consistent $p$-efficient instance $I$, we have the following results:
    \begin{equation*}
        D_I(\varphi_{AV}) \leq 
        \begin{cases}
        +\infty & \text{for } p  \in \{0, 1\}\\
        \frac{1-p}{p} & \text{for }p \in (0; \frac{1}{4}]\\
        3 & \text{for } p\in [\frac{1}{4};\frac{1}{2}]\\
        \frac{2-p}{1-p} & \text{for }p \in [\frac{1}{2};1)\text{.}
        \end{cases}
    \end{equation*}
   The above function is depicted in \Cref{fig:approval_voting_global_upper_bound}.
\end{theorem}

All these bounds are attained for instances in $\mathbb{E}^1$. While we omit the formal proof of this statement, in order to give the reader a better intuition, we illustrate hard instances for different values of $p$ in \Cref{fig:distInf}. 

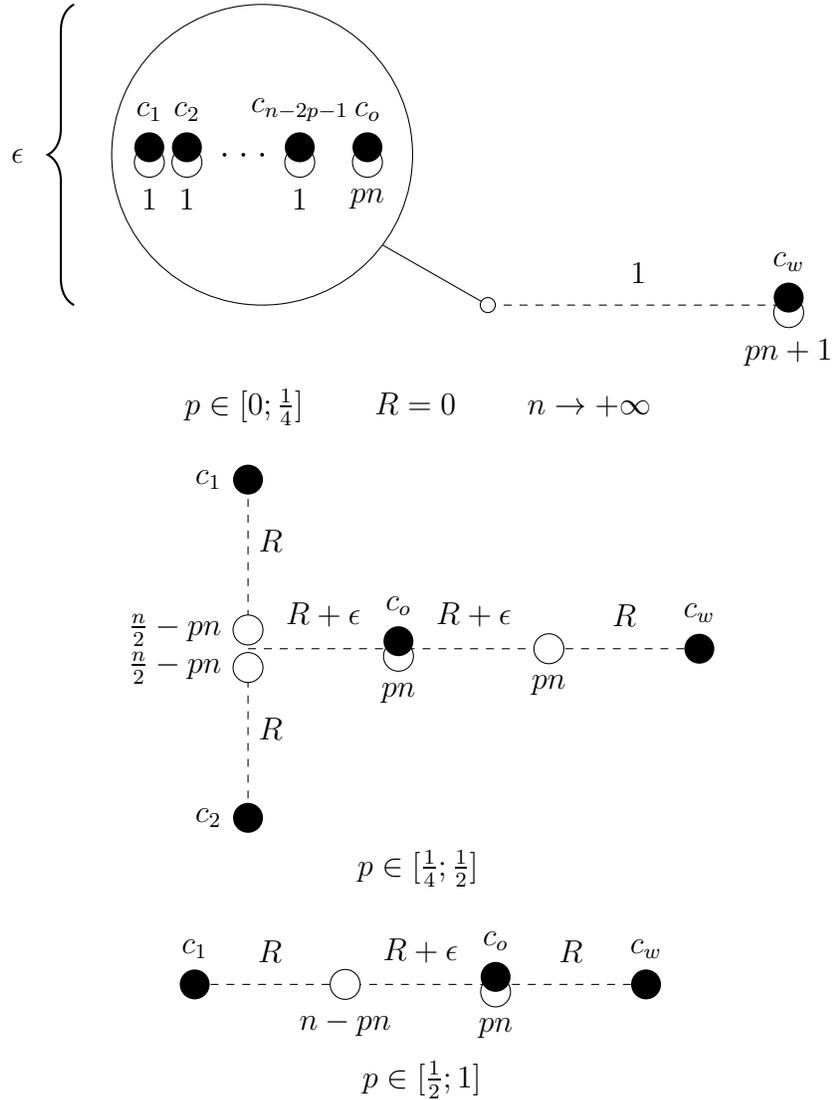
\begin{figure}[tp]
\centering
\begin{tikzpicture}
\node[label=$1$] at (2, 0) {};

\draw[style=solid] (-3.5,2) -- (0,0) {};
\draw[style=dashed] (0,0) -- (4,0) {};

\node[circle, fill=white, draw=black, minimum size=40mm] at (-3,2) {};
\draw[thick, decorate, decoration={brace, amplitude=10pt, mirror}] (-5.5,4) -- (-5.5, 0) {};

\node[label=$\epsilon$] at (-6.25, 1.6) {};

\node[circle, fill=black, label=$c_w$] at (4,0.1) {};

\node[circle, fill=black, label=$c_1$] at (-4.5,2.1) {};
\node[circle, draw=black, label=below:$1$] at (-4.5,1.9) {};

\node[circle, fill=black, label=$c_2$] at (-4,2.1) {};
\node[circle, draw=black, label=below:$1$] at (-4,1.9) {};

\node[circle, fill=black, minimum size=0.5mm, inner sep=0] at (-3.5,2) {};
\node[circle, fill=black, minimum size=0.5mm, inner sep=0] at (-3.25,2) {};
\node[circle, fill=black, minimum size=0.5mm, inner sep=0] at (-3,2) {};

\node[circle, fill=black, label=$c_{n-2p-1}$] at (-2.5,2.1) {};
\node[circle, draw=black, label=below:$1$] at (-2.5,1.9) {};

\node[circle, fill=black, label=$c_o$] at (-1.6,2.1) {};
\node[circle, draw=black, label=below:$pn$] at (-1.6,1.9) {};

\node[circle, fill=white, draw=black, minimum size=2mm, inner sep=0] at (0,0) {};
\node[circle, draw=black, label=below:$pn+1$] at (4,-0.1) {};

\end{tikzpicture}

$p \in [0; \frac{1}{4}]$ \qquad $R = 0$ \qquad $n \to +\infty$
\bigskip

\begin{tikzpicture}

\draw[style=dashed] (0,0) -- (6,0) {};
\draw[style=dashed] (0,0.2) -- (0,2.25) {};
\draw[style=dashed] (0,-0.2) -- (0,-2.25) {};

\node[label=$R$] at (0.3, -1.5) {};
\node[label=$R$] at (0.3, 1) {};

\node[label=$R+\epsilon$] at (1, 0) {};
\node[label=$R+\epsilon$] at (3, 0) {};
\node[label=$R$] at (5, 0) {};

\node[circle, fill=black, label=left:$c_1$] at (0,2.25) {};
\node[circle, fill=black, label=left:$c_2$] at (0,-2.25) {};
\node[circle, fill=black, label=$c_o$] at (2,0.1) {};
\node[circle, fill=black, label=$c_w$] at (6,0) {};

\node[circle, fill=white, draw=black, label=left:$\frac{n}{2}-pn$] at (0,0.25) {};
\node[circle, fill=white, draw=black, label=left:$\frac{n}{2}-pn$] at (0,-0.25) {};
\node[circle, draw=black, label=below:$pn$] at (2,-0.1) {};

\node[circle, fill=white, draw=black, label=below:$pn$] at (4,0) {};

\end{tikzpicture}

$p \in [\frac{1}{4}; \frac{1}{2}]$
\bigskip

\begin{tikzpicture}

\draw[style=dashed] (-2,0) -- (4,0) {};

\node[label=$R$] at (-1, 0) {};
\node[label=$R+\epsilon$] at (1, 0) {};
\node[label=$R$] at (3, 0) {};

\node[circle, fill=black, label=$c_1$] at (-2,0) {};
\node[circle, fill=black, label=$c_o$] at (2,0.1) {};
\node[circle, fill=black, label=$c_w$] at (4,0) {};

\node[circle, fill=white, draw=black, label=below:$n-pn$] at (0,0) {};
\node[circle, draw=black, label=below:$pn$] at (2,-0.1) {};

\end{tikzpicture}

$p \in [\frac{1}{2};1]$

\caption{Instances achieving the bounds given in \Cref{approval_voting_global_upper_bound}. White points correspond to groups of voters, black points---to the candidates. Here, $c_w$ is the winner of the election and $c_o$ is the optimal candidate, $c_w \succcurlyeq_{lex} c_o \succcurlyeq_{lex} c_1 \succcurlyeq_{lex} c_2 \succcurlyeq_{lex} \ldots$. $R$ is the length of the acceptability radius. Since $R$ is common for all the voters, the instances are globally consistent.}
\label{fig:distInf}
\end{figure}

Finally, for completeness, we give an analogue of \Cref{prop:approval_general_best_distortion}, but for globally-consistent instances. 

\begin{proposition}\label[proposition]{prop:global_consistent_optimistic_av}
For each instance $I \in \mathbb{I}$, there exists an approval profile $p$ globally consistent with $I$, such that
\begin{equation*}
   \frac{ \sum_{i \in N} d(i, \varphi(p))}{ \sum_{i \in N} d(i, c_o)} \leq \frac{11}{3} \text{.}
\end{equation*}
\end{proposition}

\section{AB-Distortion of Ranking Rules}
Recall that the ab-distortion of a voting rule is a value from $[0;1]$, proportional to the difference between the number of voters accepting the optimal candidate and the number of voters accepting the winner. By definition, this value equals $0$ for AV (provided the approval profile is truthful). In this section we analyze the ab-distortion of ranking-based rules. 

We start by proving the lower bound on the ab-distortion of any ranking-based voting rule.

\begin{theorem}\label[theorem]{thm:smith_lower_bound}
For each $\ell \in \mathbb{N}$ and each ranking-based rule $\varphi$, there exists a globally consistent instance $I$ such that:
\begin{enumerate}
    \item the size of the Smith set in the ranking-based profile induced by $I$ equals $\ell$,
    \item $D_I(\varphi) = 
        \begin{cases}
        \frac{\ell-1}{\ell} & \text{for }\ell \geq 2\\
        \frac{1}{2} & \text{for } \ell = 1.
        \end{cases}$
\end{enumerate}
\end{theorem}

In the subsequent part of this section we will assess the distortion of specific voting rules, specifically looking for one that meets the lower-bound from \Cref{thm:smith_lower_bound}.

\subsection{Condorcet Rules}
We start  by looking at Condorcet consistent rules. Note that the lower bound found in \Cref{thm:smith_lower_bound} is promising, as it depends on the size of the Smith set. In particular, if $\ell=1$, this bound equals $\frac{1}{2}$. Our first goal is to determine, whether Condorcet rules meet this bound.

\begin{theorem}\label{Condorcet_upper_bound}
Let $I$ be an instance where a Condorcet winner exists. Then, for each Condorcet consistent rule $\varphi$ we have $D_I(\varphi) \leq \frac{1}{2}$. This bound is achievable for a globally consistent $I\in \mathbb{E}^2$.
\end{theorem}

From the above theorem, we get that for $\ell = 1$ each Condorcet election method matches the lower bound from \Cref{thm:smith_lower_bound}. Now we will prove that there exists election rules, namely Ranked Pairs and the Schulze's rule, which match this bound for each $\ell$.

\begin{theorem}\label[theorem]{Ranked_pairs_upper_bound}
For each election instance $I$, the ab-distortion of Ranked Pairs and the Schulze's method is equal to:
    \begin{itemize}
        \item $\frac{\ell-1}{\ell}$ for $\ell \geq 2$,
        \item $\frac{1}{2}$ for $\ell = 1$,
    \end{itemize}
where $\ell$ is the size of the Smith set of $I$.
\end{theorem}

As we can see, there is no rule with a better ab-distortion than these two rules. Yet, it is not a feature of all the Condorcet methods. As we will see, even for the well-known Copeland's rule, the possible pessimistic distortion is much worse.

\begin{theorem}\label[theorem]{thm:distortion_copeland}
    For each $\epsilon > 0$, there exists a globally consistent instance $I \in \mathbb{E}^2$ for which the ab-distortion of the Copeland's rule exceeds $1-\epsilon$.  
\end{theorem}

\subsection{Scoring Rules}
Let us now move to positional scoring rules. Here, we obtain significantly worse results than for Ranked Pairs and the Schulze's rule. A general tight upper bound for the ab-distortion of any scoring rule remains an open problem. Below we provide bounds that are tight for certain specific scoring rules.

\begin{theorem}\label{scoring_rules_upper_bound}
    For a scoring rule $\varphi$ defined by vector $\vec{s}=(s_1,\ldots,s_m)$ the ab-distortion of $\varphi$ satisfies:
    \begin{enumerate}
        \item $D_I(\varphi) = 1$, if $s_1 = \ldots = s_m$,
        \item $D_I(\varphi) \leq \frac{\max_{i,j}|s_i-s_j|}{\max_{i,j}|s_i-s_j| +  \min_{i,j}|s_i-s_j|}$, otherwise.
    \end{enumerate}
\end{theorem}

The bound obtained in \Cref{scoring_rules_upper_bound} is not tight in general. For example, for Plurality we have a tighter estimation.

\begin{theorem}\label[theorem]{thm:plurality_distortion}
The ab-distortion of Plurality is $\frac{m - 1}{m}$. This bound is achieved for globally consistent instances in $\mathbb{E}^1$.
\end{theorem}

Yet, for a number of scoring rules the bound from \Cref{scoring_rules_upper_bound} is tight. Below, we give some sufficient conditions.
\begin{proposition}\label[proposition]{scoring_rules_tight}
The bound from \Cref{scoring_rules_upper_bound} is tight for each scoring rule satisfying the following conditions:
\begin{enumerate}
    \item $s_1 \geq \ldots \geq s_m$,
    \item $\forall_{1 \leq i \leq m-1} \quad s_1 - s_2 \leq s_i - s_{i+1}$
\end{enumerate}
even for globally consistent instances in $\mathbb{E}^1$.
\end{proposition}

\Cref{scoring_rules_upper_bound,scoring_rules_tight} imply the ab-distortion for a number of scoring rules.

\begin{corollary}
There exists a globally consistent instance $I \in \mathbb{E}^1$, for which:
\begin{enumerate}
    \item the ab-distortion of k-approval is $\frac{1}{1+0}=1$,
    \item the ab-distortion of Veto is $\frac{1}{1+0}=1$,
    \item the ab-distortion of Borda is $\frac{m-1}{m-1+1}=\frac{m-1}{m}$.
\end{enumerate}
\end{corollary}

\subsection{Iterative rules}

All scoring rules that we considered have poor ab-distortion, and in particular are considerably worse than Condorcet rules (especially for instances with Condorcet winners). 

Interestingly, STV in terms of acceptability, behaves worse even than Plurality. This is somehow surprising since for distance-based distortion, STV is better than any positional scoring rules, and only slightly worse than Condorcet rules. 

\begin{theorem}\label{STV_upper_bound}
The ab-distortion of STV is $\frac{2^{m-1}-1}{2^{m-1}}$.
\end{theorem}

The above bound is tight even in one-dimensional Euclidean spaces. It is also tight if we restrict ourselves to global consistent instances. There, the hard instances that we found use $(m-2)$-dimensional Euclidean space.

\begin{proposition}\label[proposition]{prop:tightness_stv}
The bound from \Cref{STV_upper_bound} is tight for locally consistent instances from $\mathbb{E}^{1}$ and globally consistent instances from $\mathbb{E}^{m-2}$.
\end{proposition}

\section{Related Work}
The spatial model of preferences is quite popular in the social choice and political science literature. For example seminal works studying spacial models we refer the reader to~\cite{DavisHinich66,Plott1967,enelow1984spatial,enelow1990advances,mckelvey1990,merrill1999unified,schofield2007spatial}.

The concept of distortion was first introduced by 
\citet{pro-ros:c:distortion}. In their work they did not assume the existence of a metric space, but rather used a generic cardinal utility model (where the voters can have arbitrarily utilities for candidates). This model was later studied by 
\citet{car-pro:j:embeddings} and 
\citet{Bou-etal}. Recently, \citet{ben-pro-qui:distortion_welfare_functions} introduced the concept of distortion for social welfare functions, i.e., functions mapping voters preferences to rankings over candidates, and \citet{ben-nat-pro-sha:participatory_budgeting_elicitation} adapted and used the concept of distortion in the context of participatory budgeting to evaluate different methods of preference elicitation. 
The studies of the concept of distortion in metric spaces were initiated by~\citet{ans-bha-elk-pos-sko:j:distortion}, and then continued by~\citet{AP16}, \citet{FFG}, \citet{GKM16}, and \citet{GAX17}.

The analysis of the distortion forms a part of a broader trend in social choice stemming from the utilitarian perspective. For classic works in welfare economics that discuss the utilitarian approach we refer the reader to the article of \citet{ng1997case} and the book of \citet{roemer1998theories}. This approach has also recently received a lot of attention from the computer science community. Apart from the papers that directly study the concept of distortion that we discussed before, examples include the works of \citet{filos2013truthful}, \citet{branzei2013bad}, and \citet{chakrabarty2014welfare}.     

\section{Conclusion}

In this paper we have extended the concept of distortion of voting rules to approval-based preferences. This extension allows to compare rules that take different types of input: approval sets and rankings over the candidates. To the best of our knowledge, only very few formal methods are known that allow for such a comparison. We are aware of only one work that formally relates these two models: \citet{Laslier2010} proved that in the strong Nash equilibrium Approval Voting selects the Condorcet winner, if such exists.    

Our contribution is twofold. First, we have determined the distortion of Approval Voting, and explained how this distortion depends on voters' approval sets. We have shown that the socially best outcome is obtained when voters approve not too many and not too few candidates. If the lengths of voters' acceptability radiuses are all equal, the best distortion is obtained when the approval sets are such that between $\frac{1}{4}$ and $\frac{1}{2}$ of the voters approve the optimal candidate.  

Second, we have defined a new concept of acceptability-based distortion (ab-distortion). Here, we assume that the voters have certain acce\-ptability thresholds; the ab-distortion of a given rule $\varphi$ measures how many voters (in the worst-case) would be satisfied from the outcomes of $\varphi$. We have determined the ab-distortion for a number of election rules (our results are summarized in \Cref{tab:summary}), and reached the following conclusions. The analysis of the classic and the acceptability-based distortions both suggest that Condorcet rules perform better than scoring and iterative ones. Further, our acceptability-based approach suggests that Ranked Pairs and the Schulze's rule are particularly good rules, in particular significantly outperforming the Copeland's rule. Thus, our study recommends Ranked Pairs or the Schulze's method as rules that robustly perform well for both criteria (total distance, and acceptability). The question whether they are the only natural ranking-based rules performing well for both criteria is open. Approval Voting is also a very good rule that can be considered an appealing alternative to them, provided the sizes of the approval sets of the voters are appropriate.  

\subsubsection*{Acknowledgments}
The authors were supported by the Foundation for Polish Science within the Homing programme (Project title: "Normative Comparison of Multiwinner Election Rules").

\bibliographystyle{named}
\bibliography{main}

\appendix
\section{Proofs Omitted from the Main Text}

\subsection{Proof of \Cref{prop:approval_general_worst_distortion}}

\begin{repproposition}{prop:approval_general_worst_distortion}
There exists an instance $I \in \mathbb{E}^1$ such that $D_I(\varphi_{AV}) = +\infty$.
\end{repproposition}

\begin{proof}
Consider the instance $I$ from \Cref{fig:approval}. We have two candidates $C=\{c_1, c_2\}$, $c_2 \succcurlyeq_{lex} c_1$, and $n$ voters. The voters are identical---for each $i \in N$ we have $d(i, c_1) = 0$, and $d(i, c_2) = 1$ (thus $c_1 \succcurlyeq_i c_2$), and they all approve all the candidates.

In $I$, $c_1$ is the optimal candidate, yet Approval Voting picks $c_1$ and $c_2$ which, together with the fact that $c_2 \succcurlyeq_{lex} c_1$, implies that $c_2$ is the winner. Thus, we get that $D_I(\varphi_{AV}) = \frac{n}{0} = +\infty$. 
\end{proof}

\begin{figure}[btp]
\centering
\begin{tikzpicture}
\node[label=$1$] at (1, 0) {};
\draw[style=dashed] (0,0) -- (1,0) -- (2, 0) {};

\node[circle, fill=black, label=$c_1$] at (0,0.1) {};
\node[circle, fill=black, label=$c_2$] at (2,0) {};

\node[circle, draw=black, label=below:$n$] at (0,-0.1) {};

\end{tikzpicture}
    \caption{Illustration of the hard instance used in the proof of \Cref{prop:approval_general_worst_distortion}. The white point indicates the position of all the voters, and black points correspond to the candidates. The length of each acceptability radius is 1---as it is the same for all the voters, the instance is globally consistent.}
\label{fig:approval}
\end{figure}
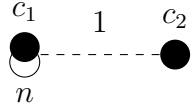

\subsection{Proof of \Cref{prop:approval_general_best_distortion}}

\begin{repproposition}{prop:approval_general_best_distortion}
For each instance $I \in \mathbb{I}$, there is an approval based profile $p$ consistent with $I$ such that $\varphi(p)$ is the optimal candidate (minimizing the total distance to voters).
\end{repproposition}

\begin{proof}
Consider an instance $I$, and let $c_o$ be an optimal candidate in $I$. Consider the following approval-based profile consistent with $I$: each voter approves $c_o$ and all the candidates more preferred to $c_o$, but does not approve any candidate less preferred than $c_o$. Candidate $c_o$ gets $n$ votes. Thus, $c_o$ will be the winner, unless some other candidate, call it $c$, also received $n$ votes and is preferred by the tie-breaking rule. If this is the case, then $c$ must Pareto dominate $c_o$, which means that $c$ is also an optimal candidate. This completes the proof.
\end{proof}

\subsection{Proof of \Cref{approval_voting_global_upper_bound}}

In the proof we will also use the following simple inequality:
\begin{lemma}\label{eqdec_lemma}
For each positive numbers $a,b,c,d$ such that $a \geq b, c \geq d$ we have that:
\begin{equation*}
     \frac{a + c}{b + c} \leq \frac{a + d}{b + d}
\end{equation*}
\end{lemma}
\begin{proof}
For each positive numbers $a,b,c,d$ such that $a \geq b, c \geq d$, we have:
\begin{align*}
    &0 \leq (a - b)(c - d) \iff \\
    &ad + bc \leq ac + bd \iff \\
    &ab + ad + bc + cd \leq ab + ac + bd + cd \iff \\
    &(a+c)(b+d) \leq (a+d)(b+c) \iff \\
    &\frac{a + c}{b + c} \leq \frac{a + d}{b + d}
\end{align*}
\end{proof}

\begin{reptheorem}{approval_voting_global_upper_bound}
    For each globally consistent $p$-efficient instance $I$, we have the following results:
    \begin{equation*}
        D_I(\varphi_{AV}) \leq 
        \begin{cases}
        +\infty & \text{for } p  \in \{0, 1\}\\
        \frac{1-p}{p} & \text{for }p \in (0; \frac{1}{4}]\\
        3 & \text{for } p\in [\frac{1}{4};\frac{1}{2}]\\
        \frac{2-p}{1-p} & \text{for }p \in [\frac{1}{2};1)\text{.}
        \end{cases}
    \end{equation*}
   The above function is depicted in \Cref{fig:approval_voting_global_upper_bound}.
\end{reptheorem}

\begin{proof}
Let $I$ be a globally consistent $p$-efficient instance. Assume that $p \notin \{0,1\}$ (otherwise, the upper bound $+\infty$ is obtained directly from the definition of distance-based distortion).
Let $c_o$ and $c_w$ denote, respectively, the optimal candidate in $I$ and the winner returned by Approval Voting. As $I$ is globally consistent, there exists $r \in \reals$ which is the length of acceptability radiuses of all the voters.

We first provide a few basic inequalities, which will be used in the further part of the proof. We will refer to these inequalities using their numbers---this will make the steps of our reasoning transparent. 
As $c_w$ is the winner of the voting, we have:
\begin{equation}\label{eq1}
    pn = |\voters{o}| \leq |\voters{w}|
\end{equation}
From the definition of the voting radius:
\begin{equation}\label{eq3}
    \forall S \subseteq \voters{o}^\complement \quad |S|r \leq \sumv{S}{o}
\end{equation}
\begin{equation}\label{eq40}
    \forall S \subseteq \voters{w} \quad |S|r \geq \sumv{S}{w}
\end{equation}
\begin{equation}\label{eq4}
    \forall S \subseteq N \quad 0 \leq \sumv{S}{o}
\end{equation}
From trivial set properties:
\begin{equation}\label{eq7}
    |\voters{o}| + |\voters{w}| - |\voters{o} \cap \voters{w}| + |(\voters{o} \cup \voters{w})^\complement| = n
\end{equation}
\begin{align}\label{eq8}
    |(\voters{o} \cup \voters{w})^\complement| &\overset{\eqref{eq7}}{=} n - |\voters{o}| - |\voters{w}| + |\voters{o} \cap \voters{w}| \nonumber\\
    &\overset{\eqref{eq1}}{\leq} n - 2pn + |\voters{o} \cap \voters{w}|
\end{align}

\begin{equation}\label{voteForBothInVoteForO}
    |\voters{o} \cap \voters{w}| \leq |\voters{o}|
\end{equation}
\\~\\
From the triangle inequality:

\begin{equation}\label{triangleInequality}
    \forall v \in N \quad \distancev{w} \leq \distancev{o} + \distance
\end{equation}
\begin{equation}\label{eq5}
    \forall v \in N \quad \distance \leq \distancev{o} + \distancev{w}
\end{equation}
\begin{equation}\label{eq6}
    \forall v \in N \quad \distance - \distancev{w} \overset{\eqref{eq5}}{\leq}  \distancev{o}
\end{equation}
\begin{equation}\label{eq600}
    \forall S \subseteq \voters{w} \quad |S|(\distance - r) \leqrefs{eq6}{eq40} \sumv{S}{o}
\end{equation}
\\~\\
From Lemma \ref{eqdec_lemma}:
\begin{equation}\label{eqdec}
     \forall a,b,c,d \in \mathbb{R}_+,a\geq b, c\geq d \quad \frac{a + c}{b + c} \leq \frac{a + d}{b + d}
\end{equation}
\begin{equation}\label{eqskip}
    \frac{a + c}{b + c} \leq \frac{a}{b}
\end{equation}
The further part of the proof will be split into three cases:
\begin{enumerate}[\quad {Case} 1]
\item $d(c_o, c_w) \leq r$, 
\item $r \leq d(c_o, c_w) \leq 2r$, and
\item $2r \leq d(c_o, c_w)$. 
\end{enumerate}
For each $p$, the final worst-case distortion is the maximum of the worst-case distortions in all these three cases.

\paragraph{Analysis of Case 1.} The following inequality holds:
\begin{equation}\label{dLEQR}
    d(c_o, c_w) \leq r.
\end{equation}
In this case we have:

\begin{align}
    D_I(\varphi_{AV}) &= \frac{\sumv{N}{w}}{\sumv{N}{o}} \leqref{triangleInequality} \frac{\sumv{N}{o} + n\distance}{\sumv{N}{o}}.
\end{align}
As the numerator is greater than the denumerator (because $D_I(\varphi_{AV}) \geq 1$):
\begin{align}
    D_I(\varphi_{AV}) &\leqrefs{eq3}{eqdec} \frac{|\voters{o}^\complement|r + n\distance}{|\voters{o}^\complement|r} \leqref{dLEQR} \frac{|\voters{o}^\complement|r + nr}{|\voters{o}^\complement|r} \nonumber\\
    &= \frac{(n - pn) + n}{n - pn} = \frac{2 - p}{1 - p}.
\end{align}

\paragraph{Analysis of Case 2.} The following inequalities hold:
\begin{equation}\label{RLEQd}
    2r \geq d(c_o, c_w) \geq r.
\end{equation}
In such case, we assess the distortion as follows:

\begin{align*}
    &D_I(\varphi_{AV}) = \frac{\sumv{N}{w}}{\sumv{N}{o}} \\
    &\quad = \frac{\sumv{\voters{w}}{w} + \sumvn{\voters{w}}{w}}{\sumv{\voters{w}}{o} + \sumvn{\voters{w}}{o}} \\
    &\quad \leqref{triangleInequality} 
    \frac{\sumv{\voters{w}}{w} + \sumvn{\voters{w}}{o} + |\voters{w}^\complement|\distance}{\sumv{\voters{w}}{o} + \sumvn{\voters{w}}{o}} \\
    &\quad \leqref{eq40} \frac{|\voters{w}|r  + \sumvn{\voters{w}}{o} + |\voters{w}^\complement|\distance}{\sumv{\voters{w}}{o} + \sumvn{\voters{w}}{o}} \\
    &\quad = \frac{|\voters{w}|r  + \sumv{\voters{o} \setminus \voters{w}}{o} + \sumvn{\voters{o} \cup \voters{w}}{o} + |\voters{w}^\complement|\distance}{\sumv{\voters{w}}{o} + \sumv{\voters{o} \setminus \voters{w}}{o} + \sumvn{\voters{o} \cup \voters{w}}{o}}.
\end{align*}
As the numerator is greater than the denumerator (because $D_I(\varphi_{AV}) \geq 1$):
\begin{align*}
    &D_I(\varphi_{AV}) \leqrefs{eq4}{eqskip} \frac{|\voters{w}|r + \sumvn{\voters{o} \cup \voters{w}}{o} + |\voters{w}^\complement|\distance}{\sumv{\voters{w}}{o} + \sumvn{\voters{o} \cup \voters{w}}{o}} \\
    &\quad \leqrefs{eq3}{eqdec} \frac{|\voters{w}|r + |(\voters{w} \cup \voters{o})^\complement|r + |\voters{w}^\complement|\distance}{\sumv{\voters{w}}{o} + |(\voters{w} \cup \voters{o})^\complement|r} \\
    &\quad = \frac{|\voters{w}|r + |(\voters{w} \cup \voters{o})^\complement|r + |\voters{w}^\complement|\distance}{\sumv{\voters{w} \cap \voters{o}}{o} + \sumv{\voters{w} \setminus \voters{o}}{o} + |(\voters{w} \cup \voters{o})^\complement|r} \\
    &\quad \leqref{eq600} \frac{|\voters{w}|r + |(\voters{w} \cup \voters{o})^\complement|r + |\voters{w}^\complement|\distance}{|\voters{w} \cap \voters{o}|(\distance - r) + \sumv{\voters{w} \setminus \voters{o}}{o} + |(\voters{w} \cup \voters{o})^\complement|r} \\
    &\quad \leqref{eq3} \frac{|\voters{w}|r + |(\voters{w} \cup \voters{o})^\complement|r + |\voters{w}^\complement|\distance}{|\voters{w} \cap \voters{o}|(\distance - r) + |\voters{w} \setminus \voters{o}|r + |(\voters{w} \cup \voters{o})^\complement|r} \\
    &\quad = \frac{|\voters{w}|r + |(\voters{w} \cup \voters{o})^\complement|r + |\voters{w}^\complement|r + |\voters{w}^\complement|(\distance-r)}{|\voters{w} \cap \voters{o}|(\distance - r) + (n - |\voters{o}|)r} \\
    &\quad \leqref{eq1}\frac{nr + |(\voters{w} \cup \voters{o})^\complement|r + (n-pn)(\distance-r)}{|\voters{w} \cap \voters{o}|(\distance - r) + (n - pn)r} \\
    &\quad \leqref{eq8} \frac{nr + (n - 2pn + |\voters{o} \cap \voters{w}|)r + (n-pn)(\distance-r)}{|\voters{w} \cap \voters{o}|(\distance - r) + (n - pn)r} \\
    &\quad =\frac{2nr - 2pnr + |\voters{o} \cap \voters{w}|r + (n-pn)(\distance-r)}{|\voters{w} \cap \voters{o}|(\distance - r) + (n - pn)r} \\
    &\quad =\frac{2nr - 2pnr + |\voters{o} \cap \voters{w}|(2r - \distance) + |\voters{o} \cap \voters{w}|(\distance - r)}{|\voters{w} \cap \voters{o}|(\distance - r) + (n - pn)r} \\
    &\qquad\qquad +\frac{ (n-pn)(\distance-r)}{|\voters{w} \cap \voters{o}|(\distance - r) + (n - pn)r} \\
    &\quad \leqref{eqskip} \frac{2nr - 2pnr + |\voters{o} \cap \voters{w}|(2r - \distance) + (n-pn)(\distance-r)}{(n - pn)r} \\
    &\quad =\frac{2n-pn}{n-pn} + \frac{- pnr + |\voters{o} \cap \voters{w}|(2r - \distance) + (n-pn)(\distance-r)}{(n - pn)r} \\
    &\quad \leqrefs{voteForBothInVoteForO}{RLEQd} \frac{2-p}{1-p} + \frac{- pnr + pn(2r - \distance) + (n-pn)(\distance-r)}{(n - pn)r} \\
    &\quad =\frac{2-p}{1-p} + \frac{pn(r - \distance) + (n-pn)(\distance-r)}{(n - pn)r} \\
    &\quad =\frac{2-p}{1-p} + \frac{(n-2pn)(\distance-r)}{(n - pn)r} \\
    &\quad =\frac{2-p}{1-p} + \frac{(1-2p)(\distance-r)}{(1 - p)r} \text{.}
\end{align*}
For $p \geq \frac{1}{2}$ the second part of the final sum is not positive, so we have:

\begin{equation*}
    \frac{2-p}{1-p} + \frac{(1-2p)(\distance-r)}{(1 - p)r} \leq \frac{2-p}{1-p}.
\end{equation*}
On the other hand, for $p \leq \frac{1}{2}$ we have:
\begin{align*}
    \frac{2-p}{1-p} + \frac{(1-2p)(\distance-r)}{(1 - p)r} 
    &\leqref{RLEQd} \frac{2-p}{1-p} + \frac{(1-2p)(2r-r)}{(1 - p)r} \\
    &= \frac{2-p}{1-p} + \frac{1-2p}{1 - p} = \frac{3 - 3p}{1 - p} = 3.
\end{align*}

\paragraph{Analysis of Case 3.} The following inequality holds:
\begin{equation}\label{2RLEQd}
    2R \leq \distance.
\end{equation}
First, let us observe that in this case sets $\voters{o}$ and $\voters{w}$ are disjoint: 
\begin{equation}\label{eqdisjoint_voters_o_w}
    \voters{o} \cap \voters{w} = \emptyset
\end{equation}. 
It also holds that:
\begin{equation*}
    p \leqrefs{eq1}{eq7} \frac{1}{2}.
\end{equation*}
Then we can limit the distortion (partially analogically as in Case 2) as follows:
\begin{align}\label{eqmain3}
\begin{split}
    D_I(\varphi_{AV}) &= \frac{\sumv{N}{w}}{\sumv{N}{o}} \\
    &= \frac{\sumv{\voters{o}}{w} + \sumv{\voters{w}}{w} + \sumvn{\voters{o} \cup \voters{w}}{w}}{\sumv{\voters{o}}{o} + \sumv{\voters{w}}{o} + \sumvn{\voters{o} \cup \voters{w}}{o}} \\
    &\leqref{triangleInequality} \frac{\sumv{\voters{o}}{o} + |\voters{o}|\distance + \sumv{\voters{w}}{w}}{\sumv{\voters{o}}{o} + \sumv{\voters{w}}{o} + \sumvn{\voters{o} \cup \voters{w}}{o}} \\
    &+\frac{\sumvn{\voters{o} \cup \voters{w}}{o} + |(\voters{o} \cup \voters{w})^\complement|\distance}{\sumv{\voters{o}}{o} + \sumv{\voters{w}}{o} + \sumvn{\voters{o} \cup \voters{w}}{o}} \text{.}
\end{split}
\end{align}
As the numerator is greater than the denumerator (because $D_I(\varphi_{AV}) \geq 1$):
\begin{align*}
    D_I(\varphi_{AV}) &\leqrefs{eq4}{eqskip} \frac{|\voters{o}|\distance + \sumv{\voters{w}}{w}}{\sumv{\voters{w}}{o} + \sumvn{\voters{o} \cup \voters{w}}{o}} \\
    &+\frac{\sumvn{\voters{o} \cup \voters{w}}{o} + |(\voters{o} \cup \voters{w})^\complement|\distance}{\sumv{\voters{w}}{o} + \sumvn{\voters{o} \cup \voters{w}}{o}} \\
    &\leqrefs{eq3}{eqdec} \frac{|\voters{o}|\distance + \sumv{\voters{w}}{w}}{\sumv{\voters{w}}{o} + |(\voters{o} \cup \voters{w})^\complement|r} \\
    &+ \frac{|(\voters{o} \cup \voters{w})^\complement|r + |(\voters{o} \cup \voters{w})^\complement|\distance}{\sumv{\voters{w}}{o} + |(\voters{o} \cup \voters{w})^\complement|r} \\
    &\leqrefs{eq40}{eq600} \frac{|\voters{o}|\distance + |\voters{w}|r}{|\voters{w}|(\distance - r) + |(\voters{o} \cup \voters{w})^\complement|r} \\
    &+ \frac{|(\voters{o} \cup \voters{w})^\complement|r + |(\voters{o} \cup \voters{w})^\complement|\distance}{|\voters{w}|(\distance - r) + |(\voters{o} \cup \voters{w})^\complement|r} \\
    &\overset{\eqref{eqdisjoint_voters_o_w}}{=} \frac{(n - |\voters{w}|)\distance + (n - |\voters{o}|)r}{|\voters{w}|(\distance - r) + (n - |\voters{w}| - |\voters{o}|)r} \\
    &= \frac{(n - |\voters{w}|)\distance + (n - |\voters{o}|)r}{|\voters{w}|(\distance - 2r) + (n - |\voters{o}|)r} \\
    &\leqref{eq1} \frac{(n - pn)\distance + (n - pn)r}{pn(\distance - 2r) + (n - pn)r} \\
    &= \frac{(1 - p)(\distance + r)}{p(\distance - r) + (1 - 2p)r} \\
    &= \frac{2(1 - p)(\distance + r)}{2p(\distance - r) + (2p - 1)(\distance - r) + (1 - 2p)(\distance + r)} \\
    &= \frac{2(1 - p)}{(4p - 1)\frac{\distance - r}{\distance + r} + 1 - 2p} \text{.}
\end{align*}
Now we can easily see that:

\begin{equation*}
    \frac{1}{3} \leqref{2RLEQd} \frac{\distance - r}{\distance + r} < 1.
\end{equation*}
Hence, for $p \in [\frac{1}{4}; \frac{1}{2}]$ we can continue our calculations as follows:

\begin{equation*}
    \frac{2(1 - p)}{(4p - 1)\frac{\distance - r}{\distance + r} + 1 - 2p} \leq \frac{2(1 - p)}{(4p - 1)\frac{1}{3} + 1 - 2p} = \frac{6(1 - p)}{4p - 1 + 3 - 6p} = \frac{6(1 - p)}{2 - 2p} = 3.
\end{equation*}
For $p \in [0; \frac{1}{4}]$ we can continue the calculations in a different way:

\begin{equation*}
    \frac{2(1 - p)}{(4p - 1)\frac{\distance - r}{\distance + r} + 1 - 2p} \leq \frac{2(1 - p)}{4p - 1 + 1 - 2p} = \frac{2(1 - p)}{2p} = \frac{1 - p}{p}.
\end{equation*}

\paragraph{Summarizing the results for the particular cases.}
Finally, we have the following results:
\begin{enumerate}
    \item For $p \in \{0,1\}$: $+\infty$.
    \item For $p \in (0; \frac{1}{4}]$: $\max(\frac{2 - p}{1-p}, 3, \frac{1 - p}{p}) = \frac{1 - p}{p}$.
    \item For $p \in [\frac{1}{4}; \frac{1}{2}]$: $\max(\frac{2 - p}{1-p}, 3, 3) = 3$.
    \item For $p \in [\frac{1}{2}; 1)$: $\max(\frac{2 - p}{1-p}, \frac{2 - p}{1-p}) = \frac{2 - p}{1-p}$.
\end{enumerate}

The hard instances for different values of $p$ are illustrated in \Cref{fig:distInf}. 
\end{proof}

\subsection{Proof of \Cref{prop:global_consistent_optimistic_av}}

\begin{repproposition}{prop:global_consistent_optimistic_av}
For each instance $I \in \mathbb{I}$, there exists an approval profile $p$ globally consistent with $I$, such that
\begin{equation*}
   \frac{ \sum_{i \in N} d(i, \varphi(p))}{ \sum_{i \in N} d(i, c_o)} \leq \frac{11}{3} \text{.}
\end{equation*}
\end{repproposition}
\begin{proof}
    Let us consider a globally consistent preference approval-based profile $A$ induced by $I$ satisfying the following conditions:
    \begin{enumerate}
        \item at least $\nicefrac{n}{4}$ voters approve an optimal candidate $c_o$,
        \item the length of acceptability radiuses $R$ is the shortest that satisfies the condition above.
    \end{enumerate}
    If the number of voters approving $c_o$ is less than $\frac{n}{2}$ then $A$ is $p$-efficient for some $p\in [\frac{1}{4};\frac{1}{2}]$ and the statement is implied directly by \Cref{approval_voting_global_upper_bound}. 
     Suppose then that $p > \frac{1}{2}$. It means that there exists subsets of voters $S \subseteq N$, such that $|S| < \frac{n}{4}$ and for each $i \in R_A(c_o) \setminus S$ we have that $d(i, c_o) = R$, while for each $i \in S$ we have that $d(i, c_o) < R$.
    \\~\\
    As $R_A(c_w) \geq R_A(c_o)$, we have that there exists a voter approving both $c_o$ and $c_w$, hence $\distance \leq 2R$. Then from the triangle inequality we obtain the following result:
    \begin{align*}
        \frac{\sumv{N}{w}}{\sumv{N}{o}} &\leq \frac{\sumv{N}{o} + n\distance}{\sumv{N}{o}} = 1+\frac{n\distance}{\sumv{N}{o}} \\
                                        &= 1+\frac{n\distance}{\sumv{N\setminus S}{o} + \sumv{S}{o}} \\
                                        &\leq 1+\frac{n\distance}{|N\setminus S|R + 0} \leq 1+\frac{n\distance}{\frac{3}{4}nR} \leq 1+\frac{2nR}{\frac{3}{4}nR} = \frac{11}{3} \text{.}
    \end{align*}
\end{proof}

\subsection{Proof of \Cref{thm:smith_lower_bound}}

\begin{reptheorem}{thm:smith_lower_bound}
For each $\ell \in \mathbb{N}$ and each ranking-based rule $\varphi$, there exists a globally consistent instance $I$ such that:
\begin{enumerate}
    \item the size of the Smith set in the ranking-based profile induced by $I$ equals $\ell$,
    \item $D_I(\varphi) = 
        \begin{cases}
        \frac{\ell-1}{\ell} & \text{for }\ell \geq 2\\
        \frac{1}{2} & \text{for } \ell = 1.
        \end{cases}$.
\end{enumerate}
\end{reptheorem}

\begin{proof}
    First, let us consider the case when $\ell \geq 2$.
    Note that $\varphi$ does not depend neither on the acceptability function nor on the specific distance values in the metric space. We will now construct a set of $\ell$ instances $\mathcal{I}=\{I_1, I_2,\ldots,I_\ell\}$ such that for each instance $I_i$ the following conditions are satisfied:
    \begin{enumerate}
        \item The number of candidates $m$ equals $\ell$.
        \item The voters are divided into groups $G_{(i, 1)}, G_{(i, 2)},\ldots, G_{(i, \ell)}$---each group contains $\frac{n}{\ell}$ voters and corresponds to a single point in the metric space.
        \item The ranking-based preferences of the voters from each group are the same and form a cyclic shift of the vector $(c_1, c_2,\ldots, c_\ell)$.
        \item The voters from group $G_{(i, k)}$ have ranking\footnote{For the sake of simplicity, in the proof we sometimes use the notation $c_i, G_{(i,k)}, I_i$ also for $i,k < 1$ or for $i,k > \ell$---in these cases, we mean the intuitive modulo notation $((i-1)\text{ mod } \ell) + 1$, $((k-1)\text{ mod } \ell) + 1$.}
        \begin{align*}
            c_{i-k+1} \succcurlyeq \ldots \succcurlyeq c_\ell \succcurlyeq c_1 \succcurlyeq \ldots \succcurlyeq c_{i-k} \text{,} 
        \end{align*}
        and accept the first $k$ candidates from the above ranking (hence, all the voters accept $c_i$ and only voters from group $G_{(i, \ell)}$ accept $c_{i+1}$).
    \end{enumerate}
    For each instance we have $\nicefrac{n}{\ell}$ voters with preferences $c_1 \succcurlyeq \ldots \succcurlyeq c_\ell$, $\nicefrac{n}{\ell}$ voters with preferences $c_\ell \succcurlyeq c_1 \succcurlyeq \ldots \succcurlyeq c_{\ell-1}$, etc.
    Hence, the ranking-based preference profile induced by each instance is the same subject to the permutation of the voters. Consequently, w.l.o.g., we can assume that in each instance the winner elected by $\varphi$ is the same. Besides, one can easily verify that in each instance all the candidates are in the Smith set.

    Suppose that $\varphi$ elects candidate $c_i$. Then we take instance $I_{i-1}$ as a witness---here, the optimal candidate is $c_{i-1}$, acceptable for $n$ voters, and candidate $c_i$ is acceptable only for $\nicefrac{n}{\ell}$ voters from $G_{(i-1,\ell)}$. Thus, for this instance we obtain the ab-distortion of $1 - \nicefrac{1}{\ell} = \nicefrac{\ell-1}{\ell}$. 
    To complete the proof for the case when $\ell \geq 2$, we need to show that it is possible to construct the instances from $\mathcal{I}$ using globally consistent acceptability functions.

    Each instance $I_i$ can be constructed through the following procedure:
    \begin{enumerate}
        \item Put all the candidates in a metric space $\mathbb{R}^{\ell-1}$ so that they are vertices of a regular ($\ell$-1)-simplex with all edges equal to 1.
        \item Set the acceptability radius of all the voters to the length of the circumradius of the simplex; let $R$ denote the length of this radius.
        \item For each $k$, we locate $\nicefrac{n}{\ell}$ voters from $G_{(i,k)}$ in the following way. Consider the subset of $k$ candidates acceptable for these voters. These candidates are also vertices of a $(k-1)$-simplex. Put the voters from $G_{(i,k)}$ in the circumcenter of this $(k-1)$-simplex.
    \end{enumerate}
    The construction for $\ell=3$ is illustrated in \Cref{fig:instances}.

    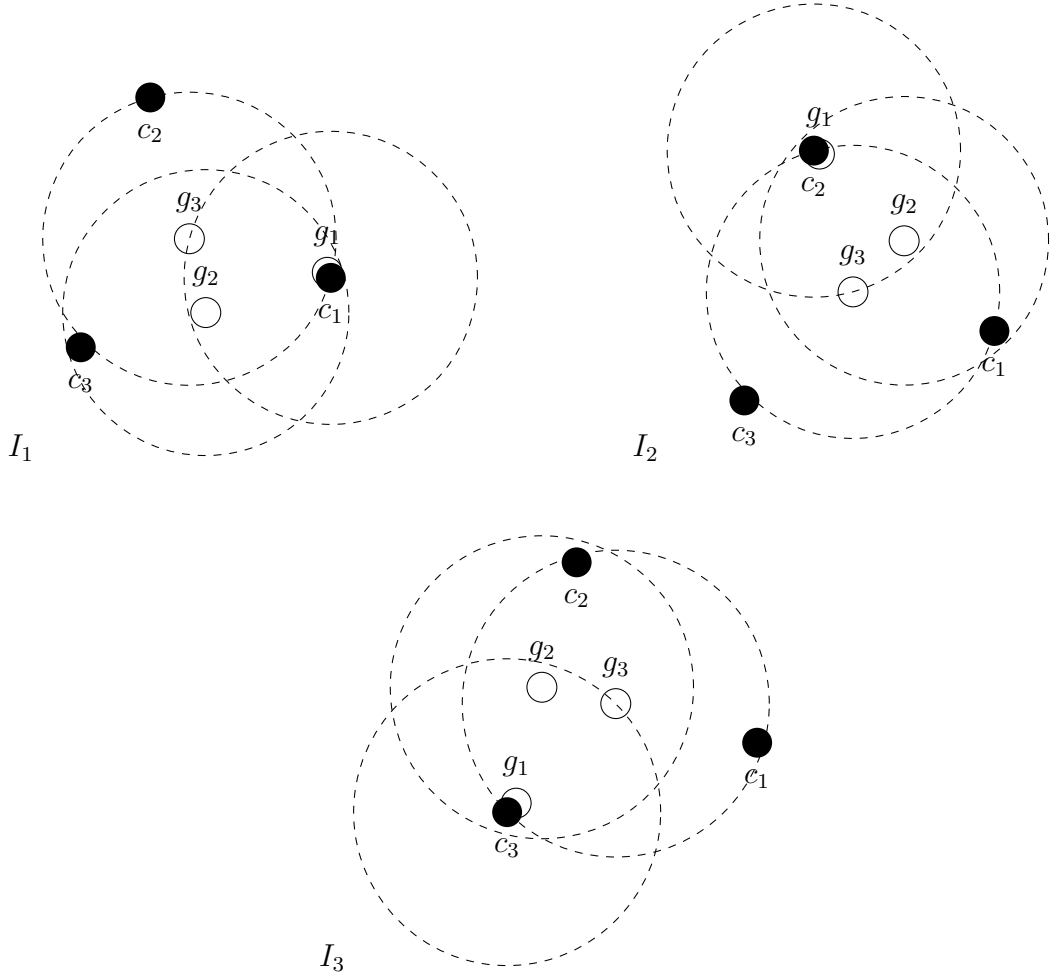
\begin{figure}[tbh]
    \begin{minipage}[t]{0.49\textwidth}
    \centering
    $I_1$
    \begin{tikzpicture}[scale=0.6]
    \node[circle, fill=black, label=below:$c_1$] at (4,0,0) {};
    \node[circle, fill=black, label=below:$c_2$] at (0,4,0) {};
    \node[circle, fill=black, label=below:$c_3$] at (0,0,4) {};

    \node[circle, draw=black, label=$g_1$] at (4,0.2,0.2) {};
    \node[circle, draw=black, label=$g_3$] at (1.41,1.41,1.41) {};
    \node[circle, draw=black, label=$g_2$] at (2,0,2) {};
    
    \node[circle, draw=black, style=dashed] at (4,0,0) [circle through={(1.41,1.41,1.41)}] {};
    \node[circle, draw=black, style=dashed] at (1.41,1.41,1.41) [circle through={(4,0,0)}] {};
    \node[circle, draw=black, style=dashed] at (2,0,2) [circle through={(2,3.17,2)}] {};
    \end{tikzpicture}
    \end{minipage}~\begin{minipage}[t]{0.49\textwidth}
    \centering
    $I_2$
    \begin{tikzpicture}[scale=0.6]
    \node[circle, fill=black, label=below:$c_1$] at (4,0,0) {};
    \node[circle, fill=black, label=below:$c_2$] at (0,4,0) {};
    \node[circle, fill=black, label=below:$c_3$] at (0,0,4) {};

    \node[circle, draw=black, label=$g_1$] at (0.2,4,0.2) {};
    \node[circle, draw=black, label=$g_3$] at (1.41,1.41,1.41) {};
    \node[circle, draw=black, label=$g_2$] at (2,2,0) {};
    
    \node[circle, draw=black, style=dashed] at (0,4,0) [circle through={(1.41,1.41,1.41)}] {};
    \node[circle, draw=black, style=dashed] at (1.41,1.41,1.41) [circle through={(0,4,0)}] {};
    \node[circle, draw=black, style=dashed] at (2,2,0) [circle through={(0,4.5,0)}] {};
    \end{tikzpicture}
    \end{minipage}
    \\~\\

    \centering
    \begin{minipage}[t]{0.49\textwidth}
    \centering
    $I_3$
    \begin{tikzpicture}[scale=0.6]
    \node[circle, fill=black, label=below:$c_1$] at (4,0,0) {};
    \node[circle, fill=black, label=below:$c_2$] at (0,4,0) {};
    \node[circle, fill=black, label=below:$c_3$] at (0,0,4) {};

    \node[circle, draw=black, label=$g_1$] at (0.2,0.2,4) {};
    \node[circle, draw=black, label=$g_3$] at (1.41,1.41,1.41) {};
    \node[circle, draw=black, label=$g_2$] at (0,2,2) {};
    
    \node[circle, draw=black, style=dashed] at (0,0,4) [circle through={(1.41,1.41,1.41)}] {};
    \node[circle, draw=black, style=dashed] at (1.41,1.41,1.41) [circle through={(0,0,4)}] {};
    \node[circle, draw=black, style=dashed] at (0,2,2) [circle through={(0,4.5,0)}] {};
    \end{tikzpicture}
    \end{minipage}
    \caption{The construction of $\mathcal{I}$ for $\ell=3$. White points correspond to the groups of voters, black ones correspond to the candidates. Circles mean acceptability spheres of the voters---as their length is common for all the voters, instances are globally consistent}
    \label{fig:instances}
    \end{figure}

    Now we will show that for each $i,k$, the construction of $G_{(i,k)}$ is correct---we need to verify the following conditions:
    \begin{enumerate}[\quad Cond 1:]
        \item The distance from each candidate $c_{i-k+1}, \ldots, c_i$ to $G_{(i,k)}$ does not exceed $R$,\label[condition]{cond1}
        \item The distance from any other candidate to $G_{(i,k)}$ exceeds $R$, \label[condition]{cond2}
        \item Ranking $c_{i-k+1} \succcurlyeq \ldots \succcurlyeq c_\ell \succcurlyeq c_1 \succcurlyeq \ldots \succcurlyeq c_{i-k}$ is consistent with the metric space for the voters from $G_{(i,k)}$.\label[condition]{cond3}
    \end{enumerate}
    Recall that for any regular $k$-simplex with length of edges equal to $1$ the length of the circumradius is equal to $\sqrt{\frac{k}{2(k+1)}}$ and the height of the simplex is equal to $\sqrt{\frac{k+1}{2k}}$ (for $k > 0$).
    
    Directly from these formulas we have some simple properties that hold for any regular $k$-simplex and $n$-simplex, both with the length of edges equal to $1$, $k < n$:
    \begin{enumerate}[\quad Property 1:]
        \item The circumradius of the $k$-simplex is smaller than the the circumradius of the $n$-simplex,\label[property]{circumradius_smaller}
        \item The height of the $k$-simplex is greater than the circumradius of the $n$-simplex.\label[property]{height_circumradius}
    \end{enumerate}

    Proof of \Cref{cond1}: Since $R$ is the circumradius of the $(\ell-1)$-simplex and the distance from any candidate $c_{i-k+1}, \ldots, c_i$ to $G_{(i,k)}$ equals the circumradius of a $(k-1)$-simplex for $k\leq \ell$, then from \Cref{circumradius_smaller} we have that this distance does not exceed $R$.

    Proof of \Cref{cond2}:
    Consider a candidate $c_x$ outside of the set $\{c_{i-k+1}, \ldots, c_i\}$. This candidate together with candidates $c_{i-k+1}, \ldots, c_i$ are vertices of a $k$-simplex. Consider now the height of this simplex dropped from $c_x$. From the properties of regular simplexes the foot of this height is the circumcenter of the $(k-1)$-simplex with vertices $c_{i-k+1}, \ldots, c_i$. Thus, this height equals the distance from $c_x$ to $G_{(i,k)}$. From \Cref{height_circumradius} we have that this distance is greater than $R$.

    Proof of \Cref{cond3}:
    From \Cref{cond1} and \Cref{cond2} we have that candidates $c_{i-k+1}, \ldots, c_i$ are closer to $G_{(i,k)}$ than any other candidate. Hence, they need to be put at the top of the rankings of all the voters from $G_{(i,k)}$. Moreover, each such a ranking is consistent with the metric space we constructed (all candidates from $\{c_{i-k+1}, \ldots, c_i\}$ have the same distance to $G_{(i,k)}$, and similarly all the candidates outside of this set). In particular, the ranking $c_{i-k+1} \succcurlyeq \ldots \succcurlyeq c_\ell \succcurlyeq c_1 \succcurlyeq \ldots \succcurlyeq c_{i-k}$ is consistent with the metric.

Now, let us move to the case when $\ell = 1$. Here, we construct the following two instances, $I_1$ and $I_2$. In both instances we have two candidates, $c_1$ and $c_2$, placed in the one-dimensional Euclidean space in points 0 and 3, respectively. Further:
\begin{enumerate}
\item In $I_1$ we have $\nicefrac{n}{2} + 1$ voters placed in point 1, and $\nicefrac{n}{2} - 1$ voters placed in point 3. The length of the acceptability radius for all the voters is equal to $2$, hence, all voters find $c_2$ acceptable, and only $\nicefrac{n}{2} + 1$ of them approve $c_1$. 
\item In $I_2$ we put $\nicefrac{n}{2} + 1$ voters in point 0, and $\nicefrac{n}{2} - 1$ in point 2. Similarly as in the previous case, we set the length of the acceptability radius to 2. Here, $c_1$ and $c_2$ are acceptable for, respectively, $n$ and $\nicefrac{n}{2} - 1$ voters.  
\end{enumerate}
Any deterministic rule cannot distinguish $I_1$ from $I_2$, so the winner will be the same in both instances. Thus, in one of them, we will get the ab-distortion of $\nicefrac{1}{2}$.

\end{proof}

\subsection{Proof of \Cref{Condorcet_upper_bound}}

\begin{reptheorem}{Condorcet_upper_bound}
Let $I$ be an instance where a Condorcet winner exists. Then, for each Condorcet consistent rule $\varphi$ we have $D_I(\varphi) \leq \frac{1}{2}$. This bound is achievable for a globally consistent $I\in \mathbb{E}^2$.
\end{reptheorem}
\begin{proof}
Let $c_w$ be the winner and $c_o$ be the optimal candidate. Since we assumed that the Condorcet candidate exists in $I$, we have that $c_w$ weakly dominates $c_o$. Then, we have:
    \begin{equation*}
        \frac{n}{2} \leq |P(c_w, c_o)| = n - |P(c_o, c_w)|
    \end{equation*}
Thus, $|P(c_o, c_w)| \leq \frac{n}{2}$. Further, we have that:
    \begin{align*}
      &|R_\lambda(c_o)| - |R_\lambda(c_w)| \\
      &\qquad = |R_\lambda(c_o) \setminus R_\lambda(c_w)| + |R_\lambda(c_o) \cap R_\lambda(c_w)| \\
      &\qquad\qquad - |R_\lambda(c_w) \setminus R_\lambda(c_o)| - |R_\lambda(c_o) \cap R_\lambda(c_w)| \\
      &\qquad = |R_\lambda(c_o) \setminus R_\lambda(c_w)| - |R_\lambda(c_w) \setminus R_\lambda(c_o)| \\
      &\qquad \leq |R_\lambda(c_o) \setminus R_\lambda(c_w)| \leq P(c_o, c_w) \leq \frac{n}{2}.
    \end{align*}
This completes the first part of the proof.

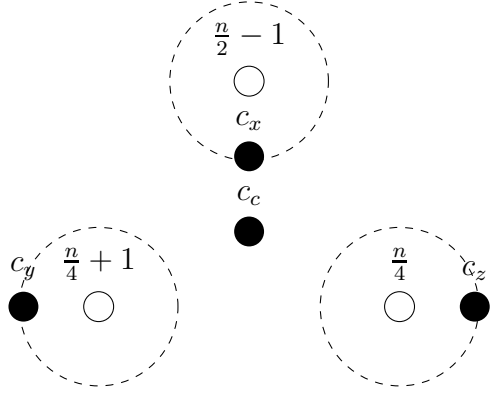
\begin{figure}[tp]
\centering
\begin{tikzpicture}
\node[circle, fill=black, label=$c_c$] at (3, 2) {};
\node[circle, fill=black, label=$c_x$] at (3, 3) {};
\node[circle, fill=black, label=$c_y$] at (0,1) {};
\node[circle, fill=black, label=$c_z$] at (6,1) {};

\node[circle, draw=black, label=$\frac{n}{2} - 1$] at (3,4) {};
\node[circle, draw=black, label=$\frac{n}{4} + 1$] at (1,1) {};
\node[circle, draw=black, label=$\frac{n}{4}$] at (5,1) {};

\node[circle, draw=black, minimum size=60pt, style=dashed] at (3,4) {};
\node[circle, draw=black, minimum size=60pt, style=dashed] at (1,1) {};
\node[circle, draw=black, minimum size=60pt, style=dashed] at (5,1) {};
\end{tikzpicture}
\caption{A hard instance witnessing the bound from \Cref{Condorcet_upper_bound} for $\ell=1$. White points correspond to groups of voters, black---to the candidates. As the length of the radius is common for all the acceptability balls, the instance is globally consistent.}
\label{fig:Condorcet}
\end{figure}

The hard instance is illustrated in \Cref{fig:Condorcet}. We have four candidates $C=\{c_x, c_y, c_z, c_c\}$. There are $\frac{n}{2} - 1$ voters with preferences $c_x \succcurlyeq c_c \succcurlyeq c_y, c_z$, approving only $c_x$, $\frac{n}{4} + 1$ voters with rankings $c_y \succcurlyeq c_c \succcurlyeq c_x, c_z$, approving only $c_y$, and $\frac{n}{4}$ voters with preferences $c_z \succcurlyeq c_c \succcurlyeq c_x, c_y$, approving only $c_z$. Candidate $c_c$ is the Condorcet winner and the optimal candidate is $c_x$. The distortion of each rule electing $c_c$ is $\frac{1}{2} - \frac{1}{n}$, which is arbitrarily close to $\frac{1}{2}$.
\end{proof}

\subsection{Proof of \Cref{Ranked_pairs_upper_bound}}

In the proof of the theorem, we will use the following definitions:
\begin{definition}
Let the \emph{immunity set} of instance $I$ be the set of candidates such that each candidate $c_a$ from this set satisfies the following condition for each $c_b \in C$: if $c_b$ dominates $c_a$, then there exists a beatpath from $c_a$ to $c_b$ which strength is greater or equal to $P(c_b, c_a)$.
\end{definition}
\begin{definition}
The election rule is \emph{immune} if for each instance it elects a candidate from the immunity set. 
\end{definition}
Note that for each instance the immunity set is the subset of the Smith set. Indeed, consider any election instance $I$. Let $c_a\in C$ belong to the immunity set of $I$ and $c_b$ belong to the Smith set of $I$. If $c_a$ weakly dominates $c_b$, then $c_a$ belongs to the Smith set. Otherwise, there exists a beatpath from $c_a$ to $c_b$. The last element of this beatpath dominates $c_b$, hence it belongs to the Smith set. If the $i$th element of the beatpath belongs to the Smith set, then so does the $i-1$th element (which dominates the $i$th one). Finally, we have that $c_a$ dominates the first element on the beatpath, hence $c_a$ belongs to the Smith set.

Now we would like to prove the following lemma:
\begin{lemma}\label[lemma]{immune_upper_bound}
The ab-distortion of each immune ranking-based election rule $\varphi$ is equal to:
    \begin{itemize}
        \item $\frac{\ell-1}{\ell}$ for $\ell \geq 2$,
        \item $\frac{1}{2}$ for $\ell = 1$,
    \end{itemize}
where $\ell$ is the size of the Smith set of considered instance.
\end{lemma}
\begin{proof}
    Let us denote the Smith set of the considered instance as $S$ and let $c_w$ be the winner according to $\varphi$. We have that:
    \begin{equation*}
        |R_\lambda(c_o)| - |R_\lambda(c_w)| \leq |R_\lambda(c_o) \setminus R_\lambda(c_w)| \leq |P(c_o, c_w)| \text{.}
    \end{equation*}
    
    In the further part of the proof we will upper bound $|P(c_o, c_w)|$. As the immunity set is a subset of the Smith set, it holds that $c_w \in S$.
    Let us consider two cases:
    \begin{enumerate}[\quad {Case} 1:]
    \item Candidate $c_w$ weakly dominates $c_o$. Then using the same reasoning as in the proof of \cref{Condorcet_upper_bound} we get the bound of $\nicefrac{1}{2}$. In this case we get the thesis, since for $\ell \geq 2$ we have that $\frac{1}{2} \leq \frac{\ell - 1}{\ell}$.
    
    \item Candidate $c_o$ dominates $c_w$. Then we have:
    \begin{equation*}
        |P(c_w, c_o)| \leq |P(c_o, c_w)|.
    \end{equation*}
    \end{enumerate}    

    We continue the analysis of Case 2. From the properties of the immunity set, there exists a vector of candidates $(c_{x_1}, c_{x_2}, \ldots, c_{c_k-2})$ for some $k \in \mathbb{N}$, such that $c_w$ dominates $c_{x_1}$, $c_{x_k-2}$ dominates $c_o$ and $\forall i \in \{1, \ldots, k-3\} \quad c_{x_i} \text{ dominates } c_{x_{i+1}}$. 
    Besides, this vector satisfies the following inequalities:
    \begin{align*}
        &|P(c_o, c_w)| \leq |P(c_w, c_{x_1})|, \\
        \forall i\in \{1, \ldots, k-3\} \quad &|P(c_o, c_w)|\leq |P(c_{x_i}, c_{x_{i+1}})|, \\
        &|P(c_o, c_w)| \leq |P(c_{x_{k-2}}, c_o)|.
    \end{align*}
    
    Summing up both sides of these $k-1$ inequalities, we have that:
    \begin{equation*}
        (k-1)|P(c_o, c_w)|
        \leq |P(c_w, c_{x_1})|+ \sum_{1 \leq i \leq k-3}{|P(c_{x_i}, c_{x_{i+1}})|} 
        + |P(c_{x_{k-2}}, c_o)|.
    \end{equation*}
    
This is equivalent to:
    \begin{align*}
        &k|P(c_o, c_w)|
        \leq |P(c_o, c_w)| + |P(c_w, c_{x_1})|
        + \sum_{1 \leq i \leq k-3}{|P(c_{x_i}, c_{x_{i+1}})|} 
        + |P(c_{x_{k-2}}, c_o)|  \iff \\
        &|P(c_o, c_w)|
        \leq \frac{|P(c_o, c_w)| + |P(c_w, c_{x_1})|
        + \sum_{1 \leq i \leq k-3}{|P(c_{x_i}, c_{x_{i+1}})|}.
        + |P(c_{x_{k-2}}, c_o)|}{k}
    \end{align*}
    
    Now let us consider the numerator in the right-hand side of the above inequality. Clearly, each of the $k$ elements of this sum can be upper-bounded by the number of the voters, $n$. However, observe that a voter $i \in N$ cannot be counted more than $k-1$ times, as her preferences are transitive and all the considered two-element vectors together make a cycle. Thus, we can upper-bound the sum by $(k-1)n$, and so we get that $|P(c_o, c_w)| \leq \frac{k-1}{k}n$.

    As $c_w \in S$ and $c_o$ dominates $c_w$, then also $c_o \in S$ and, consequently, for each $i$ we have that $c_{x_i} \in S$. Therefore, $k \leq |S|$.  Finally, we get that:
    \begin{equation*}
        \frac{|R_\lambda(c_o)|}{n} - \frac{|R_\lambda(c_w)|}{n} \leq \frac{(k-1)n}{kn} = \frac{k-1}{k} \leq \frac{|S|-1}{|S|}
    \end{equation*}

    The fact that this upper-bound is tight follows directly from \Cref{thm:smith_lower_bound}.
\end{proof}

Having this lemma proved, the proof of the main theorem is quite simple:

\begin{reptheorem}{Ranked_pairs_upper_bound}
For each election instance $I$, the ab-distortion of Ranked Pairs and the Schulze's rule is equal to:
    \begin{itemize}
        \item $\frac{\ell-1}{\ell}$ for $\ell \geq 2$,
        \item $\frac{1}{2}$ for $\ell = 1$,
    \end{itemize}
where $\ell$ is the size of the Smith set of $I$.
\end{reptheorem}
\begin{proof}
We will prove that both Ranked Pairs and the Schulze's rule are immune. Then the result will follow directly from \Cref{immune_upper_bound}.

Consider an election instance $I$, and let $c_{RP}$ be the winner according to Ranked Pairs in $I$, and $c$ be some other candidate. Assume that $c$ dominates $c_{RP}$. From the properties of Ranked Pairs, we know that the edge between $c$ and $c_{RP}$ has not been added by the election algorithm to the graph. Hence, there must have existed a path in this graph  from $c_{RP}$ to $c$, which had been added before pair $(c, c_{RP})$ was considered. Vertices on this path clearly form a beatpath from $c_{RP}$ to $c$. Besides, each edge on this path $(c_i, c_j)$ satisfies $P(c_i, c_j) \geq P(c, c_{RP})$, as it was considered by the algorithm before edge $(c, c_{RP})$. Hence, the strength of this beatpath from $c_w$ to $c$ is greater of equal to $P(c,c_{RP})$ and, consequently, $c_{RP}$ belongs to the immunity set of $I$.

Now, let us denote the winner according to the Schulze's rule by $c_S$. Assume that $c$ dominates $c_S$. Then $p[c,c_S] \geq P(c,c_S)$. On the other hand, from the properties of the Schulze's rule, we know that $p[c_S, c] \geq p[c,c_S]$. Hence, there exists a beatpath from $c_S$ to $c$ with the strength greater or equal to $p[c,c_S]$, and, consequently, also to $P(c,c_S)$. As a result, $c_S$ belongs to the immunity set of $I$.
\end{proof}

\subsection{Proof of \Cref{thm:distortion_copeland}}

\begin{reptheorem}{thm:distortion_copeland}
    For each $\epsilon > 0$, there exists a globally consistent instance $I \in \mathbb{E}^2$ for which the ab-distortion of the Copeland's rule exceeds $1-\epsilon$.  
\end{reptheorem}

\begin{proof}
    Let us fix $\epsilon > 0$, and consider the following instance $I$ with $n  > \frac{2}{\epsilon}$ voters and three candidates, $C=\{c_1, c_2, c_3\}$, $c_2 \succcurlyeq_{lex} c_1, c_3$.
    We put the candidates in $\mathbb{R}^2$ so that they are all vertices of an equilateral triangle. The length of the acceptability radius $R$ is the same for all the voters, and equals the length of the circumradius of the triangle. Finally, let us describe the positions of the voters: we put $\frac{n}{2} - 1$ of them in the same point as $c_1$ and set their preferences to $c_1 \succcurlyeq c_2 \succcurlyeq c_3$ (they approve only $c_1$). We put $\frac{n}{2} - 1$ voters in the middle of the segment between $c_1$ and $c_3$ and set their preferences to $c_3 \succcurlyeq c_1 \succcurlyeq c_2$ (they approve $c_1$ and $c_3$).
The remaining 2 voters are located in the circumcenter and have preferences $c_2 \succcurlyeq c_3 \succcurlyeq c_1$ (they approve $c_1$, $c_2$, and $c_3$).
    This instance is illustrated in \Cref{fig:Copeland}.
  
    In this instance, $c_1$ is the optimal candidate, approved by $n$ voters. However, we have that $c_1$ dominates $c_2$, $c_2$ dominates $c_3$ and $c_3$ dominates $c_1$. Hence, all the candidates are elected by the Copeland's rule and we need to use the lexicographical tie-breaking rule to choose the winner---which finally picks $c_2$, approved only by 2 voters. This gives the ab-distortion of $1 - \nicefrac{2}{n}$, which for $n  > \frac{2}{\epsilon}$ is greater than $1-\epsilon$.
\end{proof}

\begin{figure}[tbh]
\centering
    \begin{tikzpicture}[scale=0.6]
    \node[circle, fill=black, label=below:$c_1$] at (4,0,0) {};
    \node[circle, fill=black, label=below:$c_2$] at (0,4,0) {};
    \node[circle, fill=black, label=below:$c_3$] at (0,0,4) {};

    \node[circle, draw=black, label=$\frac{n}{2}-1$] at (2,0,2) {};
    \node[circle, draw=black, label=$\frac{n}{2}-1$] at (4,0.2,0.2) {};
    \node[circle, draw=black, label=$2$] at (1.41,1.41,1.41) {};
    
    \node[circle, draw=black, minimum size=190pt, style=dashed] at (2,0,2) {};
    \node[circle, draw=black, minimum size=190pt, style=dashed] at (4,0.2,0.2) {};
    \node[circle, draw=black, minimum size=190pt, style=dashed] at (1.41,1.41,1.41) {};
    \end{tikzpicture}
    \caption{The hard instance used in the proof of \Cref{thm:distortion_copeland}. White points correspond to groups of voters, black points---to candidates. Circles represent acceptability balls of the voters.}
\label{fig:Copeland}
\end{figure}
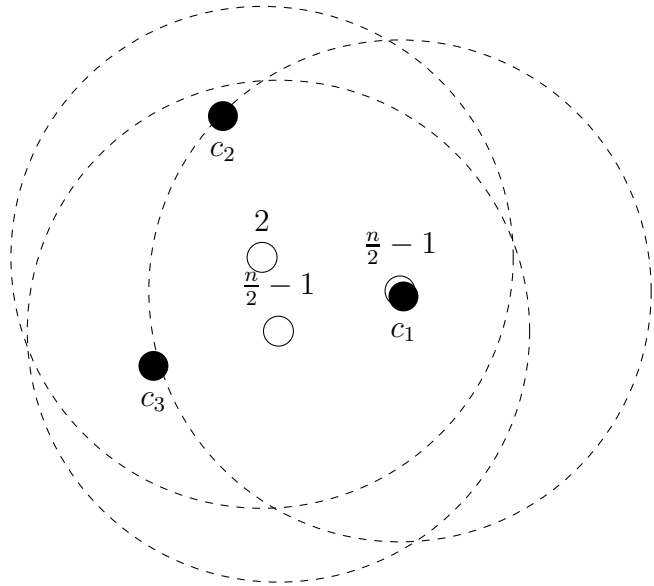

\subsection{Proof of \Cref{scoring_rules_upper_bound}}

\begin{reptheorem}{scoring_rules_upper_bound}
    For a scoring rule $\varphi$ defined by vector $\vec{s}=(s_1,\ldots,s_m)$ the ab-distortion of $\varphi$ satisfies:
    \begin{enumerate}
        \item $D_I(\varphi) = 1$, if $s_1 = \ldots = s_m$,
        \item $D_I(\varphi) \leq \frac{\max_{i,j}|s_i-s_j|}{\max_{i,j}|s_i-s_j| +  \min_{i,j}|s_i-s_j|}$, otherwise.
    \end{enumerate}
\end{reptheorem}
\begin{proof}
Let $c_w$ be the winner elected by $\varphi$ and $c_o$ be the optimal candidate. Let us consider the highest possible difference $\psi$ between the number of points gained by $c_w$ and $c_o$. Clearly $\psi \geq 0$. Besides, we know that for all voters from $P(c_o, c_w)$ candidate $c_w$ gained less points than $c_o$. Further, candidate $c_w$ could gain from each voter from this group at least $\min_{i,j}|s_i-s_j|=p$ points less than $c_o$. On the other hand, by considering the voters from $P(c_w, c_o)$ we infer that the highest possible difference in scores gained by the two candidates equals $\max_{i,j}|s_i-s_j||P(c_w, c_o)|=q|P(c_w, c_o)|$. Finally we have the following inequality:
\begin{align*}
    0 \leq \psi 
    & \leq -p|P(c_o, c_w)| + q|P(c_w, c_o)| \\
    &= -p|P(c_o, c_w)| + q(n - |P(c_o, c_w)|) \\
    &= qn - (p+q)|P(c_o, c_w)|
\end{align*}
Note that $p+q=0$ is equivalent to the condition $s_1 = \ldots = s_m$. Assume now that $p+q > 0$. In such a case we have that:
\begin{equation*}
    |P(c_o, c_w)| \leq \frac{q}{p+q}n
\end{equation*}
Further, as in the first displayed inequality in the proof of \Cref{Ranked_pairs_upper_bound}, we get that:
\begin{equation*}
    |R_\lambda(c_o)| - |R_\lambda(c_w)| \leq |R_\lambda(c_o) \setminus R_\lambda(c_w)| \leq |P(c_o, c_w)|.
\end{equation*}
Thus, the ab-distortion of $\varphi$ does not exceed $\frac{q}{p+q}$.
\end{proof}

\subsection{Proof of \Cref{thm:plurality_distortion}}

\begin{reptheorem}{thm:plurality_distortion}
The ab-distortion of Plurality is $\frac{m - 1}{m}$. This bound is achieved for globally consistent instances in $\mathbb{E}^1$.
\end{reptheorem}

\begin{proof}
    Note that the winner elected by Plurality needs to be the closest candidate for $\frac{n}{m}$ voters. Hence, from the non-emptiness and local consistency of the acceptability function, we have that $c$ is acceptable for at least $\frac{n}{m}$ voters. Thus, the distortion does not exceed $1 - \frac{1}{m} = \frac{m - 1}{m}$. This bound is tight for the instance constructed as follows:
    
Let $C=\{c_1, \ldots, c_m\}$ be the set of candidates, $c_1 \succcurlyeq_{lex} c_2 \succcurlyeq_{lex} \ldots \succcurlyeq_{lex} c_m$. We divide the voters into $m$ equal-size groups, each putting a different candidate on the top position in their preference rankings. All the candidates are acceptable for voters voting for $c_1$, and all the candidates except for $c_1$ are acceptable for the remaining voters. he illustration of a globally consistent instance in $\mathbb{E}^1$ satisfying the above conditions is presented in \Cref{fig:plurality}.

\begin{figure}[tp]
\centering
\begin{tikzpicture}
\node[label=$R-\epsilon$] at (1, 0) {};
\node[label=$R$] at (3, 0) {};
\node[label=$R$] at (5, 0) {};

\draw[style=dashed] (0,0) -- (6, 0) {};

\node[circle, fill=black, label=below:$c_1$] at (0,0) {};
\node[circle, fill=white, draw=black, label=below:$\frac{n}{m}$] at (2,0) {};

\node[circle, fill=black, label=below:{$c_2, \ldots, c_m$}] at (4,0) {};
\node[circle, fill=white, draw=black, label=below:{$\frac{n}{m},\ldots,\frac{n}{m}$}] at (6,0) {};

\end{tikzpicture}
    \caption{The hard instance used in the proof of \Cref{thm:plurality_distortion}. White points correspond to groups of voters, black---to the candidates. For the clarity of the presentation, candidates $c_2,\ldots,c_m$ and voters voting for them are associated with the same points. $R$ is the length of each acceptability radius. As this length is common for all the voters, the instance is globally consistent.}
\label{fig:plurality}
\end{figure}
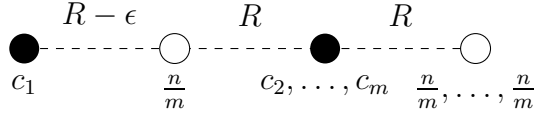

In this instance each candidate gains $\frac{m}{n}$ points and the winner is $c_1$ (due to the lexicographical tie-breaking rule). The optimal candidates are $c_2,\ldots,c_m$, and they are acceptable for $n$ voters. Hence, the distortion is $\frac{m-1}{m}$.
\end{proof}

\subsection{Proof of \Cref{scoring_rules_tight}}

\begin{repproposition}{scoring_rules_tight}
The bound from \Cref{scoring_rules_upper_bound} is tight for each scoring rule satisfying the following conditions:
\begin{enumerate}
    \item $s_1 \geq \ldots \geq s_m$,
    \item $\forall_{1 \leq i \leq m-1} \quad s_1 - s_2 \leq s_i - s_{i+1}$
\end{enumerate}
even for globally consistent instances in $\mathbb{E}^1$.
\end{repproposition}

\begin{proof}
Note that for the scoring rules satisfying the conditions of the proposition we have that $\max_{i,j}|s_i-s_j| = s_1 - s_m$ and $\min_{i,j}|s_i-s_j| = s_1 - s_2$. Let us use the same notation as in the proof of \Cref{scoring_rules_upper_bound}: $p=s_1-s_2$, $q=s_1-s_m$. 

First let us consider the case when $s_1=\ldots=s_m$, and consequently $p+q=0$. In such case, the bound $1$ is achieved by the following instance with $m$ candidates, $c_m \succcurlyeq_{lex} c_1, \ldots, c_{m-1}$.
All the voters have preferences $c_1 \succcurlyeq \ldots \succcurlyeq c_m$ and approve only $c_1$. This instance is illustrated in the upper row of \Cref{fig:equal_score}.
The optimal candidate is $c_1$, approved by $n$ voters. However, all the candidates received the same number of points, and therefore the winner is $c_m$, approved by no voter. 

In the further part of the proof we assume that $p+q>0$. Then the bound from \Cref{scoring_rules_upper_bound} is $\frac{q}{p+q}=\frac{s_1-s_m}{2s_1-s_2-s_m}$.
\begin{figure}[tp]
\centering
\begin{tikzpicture}
\node[label=$R+\epsilon$] at (2, 0) {};
\draw[style=dashed] (0,0) -- (4, 0) {};

\node[circle, fill=black, label=$c_1$] at (0,0.1) {};
\node[circle, fill=black, label={$c_2, \ldots, c_m$}] at (4,0) {};

\node[circle, draw=black, label=below:$n$] at (0,-0.1) {};

\end{tikzpicture}
\\
1. $s_1 \geq \ldots \geq s_m$
\\~\\~\\~\\
\begin{tikzpicture}

\node[label=$\epsilon$] at (1, 0) {};
\node[label=$R$] at (4, 0) {};
\node[label=$\epsilon$] at (7, 0) {};

\draw[style=dashed] (0,0) -- (8, 0) {};

\node[circle, fill=white, draw=black, label=below:$\frac{s_1-s_m}{2s_1-s_2-s_m}$] at (0,0) {};
\node[circle, fill=black, label=$c_1$] at (2,0) {};
\node[circle, fill=white, draw=black, label=below:$\frac{s_1-s_2}{2s_1-s_2-s_m}$] at (6,-0.1) {};
\node[circle, fill=black, label={$c_2$}] at (6,0.1) {};
\node[circle, fill=black, label=right:{$c_3, \ldots, c_m$}] at (8,0) {};

\end{tikzpicture}
\\~\\
2. $\forall_{1 \leq i \leq m-1} \quad s_1 - s_2 \leq s_i - s_{i+1}$
\\~\\

    \caption{The hard instances used in the proof of \Cref{scoring_rules_tight}. White points correspond to groups of voters, black points correspond to the candidates. $R$ is the length of each acceptability radius. As this length is common for all the voters, the instance is globally consistent.}
\label{fig:equal_score}
\end{figure}
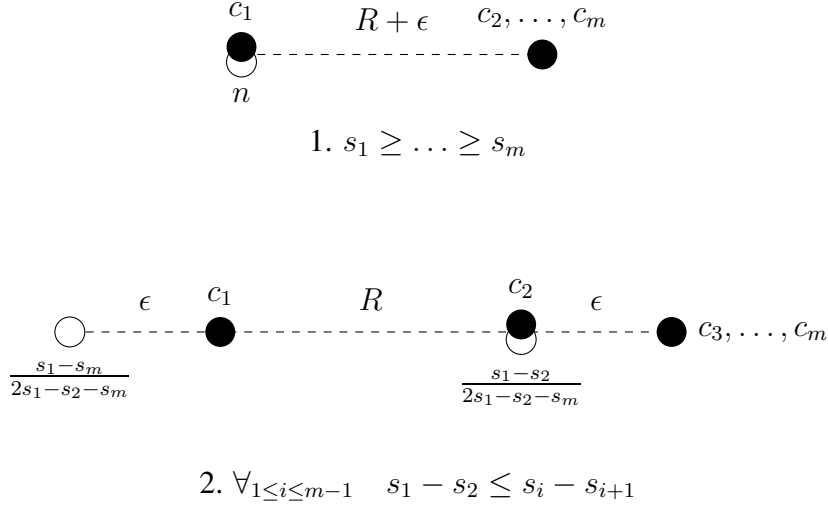

Now consider the following instance. We have $m$ candidates, $C=\{c_1,\ldots,c_m\}$, $c_2 \succcurlyeq_{lex} c_1$, and $n$ voters divided into two groups: $\frac{s_1-s_m}{2s_1-s_2-s_m}n$ of them have preferences $c_1 \succcurlyeq c_2 \succcurlyeq \ldots \succcurlyeq c_m$ and only $c_1$ is acceptable for them, and the remaining $\frac{s_1-s_2}{2s_1-s_2-s_m}n$ voters have preferences $c_2 \succcurlyeq \ldots \succcurlyeq c_m \succcurlyeq c_1$ with all the candidates being acceptable. An illustration of a globally consistent instance in $\mathbb{E}^1$ satisfying the above conditions is given in the lower row of \Cref{fig:equal_score}.

    In this instance $c_2$ gains $s_1$ points from $\frac{s_1-s_2}{2s_1-s_2-s_m}n$ voters and $s_2$ points from $\frac{s_1-s_m}{2s_1-s_2-s_m}n$ voters. The total number of points that $c_2$ received is:
        \begin{align*}
            s_1\frac{s_1-s_2}{2s_1-s_2-s_m}n &+ s_2\frac{s_1-s_m}{2s_1-s_2-s_m}n = \frac{s_1(s_1-s_2)+s_2(s_1-s_m)}{2s_1-s_2-s_m}n \\&= \frac{s_1^2-s_1s_2+s_1s_2-s_2s_m}{2s_1-s_2-s_m}n = \frac{s_1^2-s_2s_m}{2s_1-s_2-s_m}n \text{.}
        \end{align*}
   On the other hand, $c_1$ gains $s_1$ points from $\frac{s_1-s_m}{2s_1-s_2-s_m}n$ voters and $s_m$ points from $\frac{s_1-s_2}{2s_1-s_2-s_m}n$ voters. Its total score is equal to:
        \begin{align*}
            s_1\frac{s_1-s_m}{2s_1-s_2-s_m}n &+ s_m\frac{s_1-s_2}{2s_1-s_2-s_m}n = \frac{s_1(s_1-s_m)+s_m(s_1-s_2)}{2s_1-s_2-s_m}n \\&= \frac{s_1^2-s_1s_m+s_1s_m-s_2s_m}{2s_1-s_2-s_m}n = \frac{s_1^2-s_2s_m}{2s_1-s_2-s_m}n
        \end{align*}
   Since candidates $c_3,\ldots,c_m$ gain less points than $c_2$ (as they are farther away from each voter than $c_2$ and it holds that $s_1 \geq \ldots \geq s_m$), none of them can become the winner. Thus, $c_2$ is the winner of the election due to the lexicographical tie-breaking.

   The optimal candidate is $c_1$, and is acceptable for $n$ voters, while $c_2$ only for $\frac{s_1-s_2}{2s_1-s_2-s_m}n$ of them.
    Hence, the distortion is equal to $\frac{s_1-s_m}{2s_1-s_2-s_m}$.
\end{proof}

\subsection{Proof of \Cref{STV_upper_bound}}

\begin{reptheorem}{STV_upper_bound}
The ab-distortion of STV is $\frac{2^{m-1}-1}{2^{m-1}}$.
\end{reptheorem}

\begin{proof}
Let $c_w$ be the winner according to STV. Note that, the winner is the only candidate that survived to the the last round and in this round gains $n$ votes. The number of rounds equals $m$.

Note that if a candidate in some round gains $k$ points (according to Plurality), then in the next round it gains at most $2k$ points. Indeed, otherwise this candidate would need to gain at least $k+1$ points transferred from the removed candidate, so STV would not remove the candidate with the least number of points, a contradiction.
Consequently, we have that $c_w$, who has $n$ points after $m$ rounds, needs to get at least $\frac{n}{2^{m-1}}$ points in the first round. Similarly as in case of Plurality, from the non-emptiness and local consistency of the acceptability function, we have that $c_w$ is acceptable for at least $\frac{n}{2^m}$ voters. Hence, the distortion does not exceed $1 - \frac{1}{2^{m-1}}=\frac{2^{m-1}-1}{2^{m-1}}$. We will prove the tightness in \Cref{prop:tightness_stv}.
\end{proof}

\subsection{Proof of \Cref{prop:tightness_stv}}

\begin{repproposition}{prop:tightness_stv}
The bound from \Cref{STV_upper_bound} is tight for locally consistent instances from $\mathbb{E}^{1}$ and globally consistent instances from $\mathbb{E}^{m-2}$.
\end{repproposition}
\begin{proof}

\begin{figure}[tbh]
\centering
\begin{tikzpicture}
\draw[style=dashed] (0,0) -- (8, 0) {};
\node[label=1] at (1.5, 0) {};
\node[label=2] at (3, 0) {};
\node[label=4] at (6, 0) {};

\node[circle, fill=black, label=$c_m$] at (1,0.1) {};
\node[circle, draw=black, label=below:$\frac{n}{2^{m-1}}$] at (1,-0.1) {};

\node[circle, fill=black, label=$c_1$] at (2,0.1) {};
\node[circle, draw=black, label=below:$\frac{n}{2^{m-1}}$] at (2,-0.1) {};

\node[circle, fill=black, label=$c_2$] at (4,0.1) {};
\node[circle, draw=black, label=below:$\frac{n}{2^{m-2}}$] at (4,-0.1) {};

\node[circle, fill=black, label=$c_3$] at (8,0.1) {};
\node[circle, draw=black, label=below:$\frac{n}{2^{m-3}}$] at (8,-0.1) {};

\node[circle, fill=black, minimum size=0.5mm, inner sep=0] at (9,0) {};
\node[circle, fill=black, minimum size=0.5mm, inner sep=0] at (9.25,0) {};
\node[circle, fill=black, minimum size=0.5mm, inner sep=0] at (9.5,0) {};

\node[circle, fill=black, label=$c_{m-1}$] at (12,0.1) {};
\node[circle, draw=black, label=below:$\frac{n}{2}$] at (12,-0.1) {};

\end{tikzpicture}
\\~\\
1. A locally consistent instance from $\mathbb{E}^{1}$.
\\~\\
    \begin{tikzpicture}[scale=0.6]
    
    \node[circle, draw=black, label=below:$\frac{n}{2}$] at (0,4) {};
    \node[circle, draw=black, label=below:$\frac{n}{8}$] at (0,-0.1) {};
    \node[circle, draw=black, label=below:$\frac{n}{4}$] at (-3.5,-2) {};
    \node[circle, draw=black, label=below:$\frac{n}{8}$] at (3.5, -2) {};
    
    \node[circle, fill=black, label=$c_4$] at (0,8) {};
    \node[circle, fill=black, label=$c_1$] at (0,0.1) {};
    \node[circle, fill=black, label=$c_3$] at (-7,-4) {};
    \node[circle, fill=black, label=$c_2$] at (7,-4) {};

    \node[circle, draw=black, style=dashed] at (0,4) [circle through={(0,0)}] {};
    \node[circle, draw=black, style=dashed] at (0,0) [circle through={(0,4)}] {};
    \node[circle, draw=black, style=dashed] at (-3.5,-2) [circle through={(0,0)}] {};
    \node[circle, draw=black, style=dashed] at (3.5,-2) [circle through={(0,0)}] {};
    \end{tikzpicture}
\\~\\
2. A globally consistent instances from $\mathbb{E}^{m-2}$.
\\~\\
    \caption{The hard instances used in the proof of \Cref{STV_upper_bound}. White points correspond to groups of voters, black points---to the candidates. In the upper instance $c_1$ is acceptable for all the voters, and $c_m$ is acceptable only for the voters from the first group to the left.}
\label{fig:stv}
\end{figure}

First, we construct a locally consistent instance in $\mathbb{E}^1$. We have $m$ candidates $C=\{c_1,\ldots,c_m\}$, $c_m \succcurlyeq_{lex} c_1, \ldots, c_{m-1}$. The candidates $c_m, c_1, \ldots, c_{m-1}$ are associated respectively with points $1, 2, 4, 8, \ldots, 2^m$. The voters are divided into $m$ groups; the first group consists of $\frac{n}{2^{m-1}}$ voters, and for each $i \geq 2$ the $i$
-th group consists of $\frac{n}{2^{m-i+1}}$ voters. These groups also also associated, respectively, with the points $1, 2, 4, 8, \ldots, 2^m$. Candidate $c_1$ is acceptable for all the voters, and $c_m$ is acceptable only for the $\frac{n}{2^{m-1}}$ voters from the first group. This instance is depicted in the upper row of \Cref{fig:stv}.

In the first round, candidate $c_m$ gains $\frac{n}{2^{m-1}}$ points, $c_1$ gains $\frac{n}{2^{m-1}}$ points, $\ldots$, $c_{m-1}$ gains $\frac{n}{2}$ points. The removed candidate is $c_1$, due to the lexicographical tie-breaking. In the second round, all the voters who voted in the first round for $c_1$ vote for $c_m$ (as their distance to $c_m$ is $1$ while the distance to $c_2$ is $2$, and other candidates are farther than $c_2$). Hence, $c_m$ gains in total $\frac{n}{2^{m-2}}$ points and the number of points of other candidates do not change. The removed candidate is $c_2$ with $\frac{n}{2^{m-2}}$ points. In the next rounds, the removed candidates are respectively $c_3,\ldots, c_m$, and all their points are always transferred to $c_m$. The winner of the election is $c_m$, and the optimal candidate is $c_1$. Since $c_1$ is acceptable for $n$ voters, and $c_m$ only for $\frac{n}{2^{m-1}}$, we get the distortion of $\frac{2^{m-1}-1}{2^{m-1}}$.

Now, let us construct a globally consistent instance from $\mathbb{E}^{m-2}$. We have $C=\{c_1,\ldots,c_m\}$, $c_m \succcurlyeq_{lex} c_1, \ldots, c_{m-1}$, and all the candidates except for $c_1$ are vertices of a regular $(m-2)$-simplex in $\mathbb{R}^{m-2}$; $c_1$ is located in the circumcenter of the simplex---let us denote this point as $O$ and the length of the circumradius of this simplex as $R$. The voters are divided into $m$ groups ($g_1,\ldots,g_m$) as follows: $g_1$ consists of $\frac{n}{2^{m-1}}$ voters and is associated with the point $O$; for $i>1$, $g_i$ consists of $\nicefrac{n}{2^{m-i}}$ voters and is associated with the point in the middle of the segment between $c_i$ and $O$ (the distance from $g_i$ to $c_1$ equals the distance from $g_i$ to $c_i$ and equals $\nicefrac{R}{2}$). The acceptability radius for each voter is equal to $\nicefrac{R}{2}$. The example of this instance for $m=4$ is presented in the lower row of \Cref{fig:stv}.

For $i>1$, every voter from $g_i$ accepts only $c_i$ and $c_1$. For each $j \neq 1, i$, consider the triangle $\{c_i, c_j, g_i\}$. Recall that for any regular $k$-simplex with the edge length equal to  $1$ the circumradius is equal to $\sqrt{\frac{k}{2(k+1)}}$, which is strictly less than 1. From this formula we obtain that $R < d(c_i,c_j)$. Hence, from the triangle inequality we have that:
\begin{equation*}
    d(g_i,c_j) \geq d(c_i,c_j) - d(c_i,g_i) = d(c_i,c_j) - \frac{R}{2} > \frac{R}{2}
\end{equation*}
Hence, $c_j$ is not acceptable for $g_i$. Voters from $g_1$ accept only $c_1$ (as $d(g_1,c_1)=0$ and for each $i>1$ $d(g_1,c_i)=R$).

From the symmetry of the construction, we have that for each $i>1$, all the candidates except for $c_1$ and $c_i$ are in the same distance from $g_i$. Hence, we can assume that voters from $g_i$ have ranking $c_i \geq c_1 \geq c_m \geq c_2,\ldots,c_{m-1}$

STV applied to this rule gives exactly the same result as for the previous instance. We eliminate respectively candidates $c_1, c_2, \ldots, c_{m-1}$---and all the voters who voted for $c_i$ in the $i$th round, transfer their vote to $c_m$. Hence, finally $c_m$ will be selected as the winner (acceptable only for voters from $g_m$), while $c_1$ is the optimal candidate (acceptable for all the voters).
Consequently, we get the distortion of $\frac{2^{m-1}-1}{2^{m-1}}$. 
\end{proof}

\end{document}